\documentclass[11pt,reqno]{article}
\usepackage{amssymb,amsmath,amsfonts}
\usepackage{amsthm}
\usepackage{graphicx}
\usepackage{epstopdf}
\usepackage{bm}
\usepackage{paralist}
\usepackage{color}
\usepackage{fullpage}
\usepackage{algorithm}
\usepackage[noend]{algorithmic}

\newtheorem{claim}{Claim}[section]
\newtheorem{lemma}[claim]{Lemma}

\newtheorem{theorem}{Theorem}

\newtheorem{corollary}[claim]{Corollary}

\theoremstyle{definition}

\newtheorem{remark}[claim]{Remark}

\def\sign{{\rm sign}}
\def\est{\mbox{\tiny\rm  est}}
\def\tY{\widetilde{Y}}
\def\tG{\widetilde{G}}
\def\tZ{\widetilde{Z}}
\def\oxi{\overline{\xi}}
\def\dv{{\partial v}}
\def\di{{\partial i}}
\def\deg{{\rm deg}}
\def\Var{{\rm Var}}
\def\bBp{{\boldsymbol B}^{\mbox{\tiny new}}}
\def\dBet{{\rm Beta}}
\def\ind{{\mathbb I}}
\def\<{\langle}
\def\>{\rangle}
\def\eps{\varepsilon}
\def\ve{\varepsilon}
\def\bphi{{\boldsymbol \varphi}}

\def\de{{\rm d}}
\def\hx{\widehat{x}}

\def\bup{{\boldsymbol u}^{\perp}}
\def\uper{{u}^{\perp}}

\def\baR{\overline{R}}
\def\projp{{\boldsymbol P}^{\perp}}
\def\bg{{\boldsymbol g}}
\def\bx{{\boldsymbol x}}
\def\hbx{{\widehat{\boldsymbol x}}}
\def\bxz{{\boldsymbol x_0}}

\def\cG{{\mathcal G}}
\def\bD{{\boldsymbol D}}
\def\bE{{\boldsymbol E}}
\def\bM{{\boldsymbol M}}
\def\bX{{\boldsymbol X}}
\def\bP{{\boldsymbol P}}
\def\bA{{\boldsymbol A}}
\def\bAc{{\boldsymbol A}^{\mbox{\tiny cen}}}
\def\btB{{\boldsymbol{\tilde{B}}}}
\def\bB{{\boldsymbol B}}
\def\bZ{{\boldsymbol Z}}
\def\bY{{\boldsymbol Y}}

\def\bW{{\boldsymbol W}}
\def\bU{{\boldsymbol U}}
\def\bV{{\boldsymbol V}}

\def\bF{{\boldsymbol F}}
\def\bR{{\boldsymbol R}}
\def\bH{{\boldsymbol H}}
\def\Ga{{\boldsymbol J}}

\def\bx{{\boldsymbol x}}
\def\bw{{\boldsymbol w}}
\def\bv{{\boldsymbol v}}
\def\bu{{\boldsymbol u}}

\def\bfe{{\boldsymbol e}}

\def\id{{\rm I}}
\def\sT{{\sf T}}
\def\GOE{{\rm GOE}}

\def\Palg{{\mathbb P}_{\mbox{\tiny\rm alg}}}
\def\Ealg{{\mathbb E}_{\mbox{\tiny\rm alg}}}

\def\cE{{\mathcal{E}}}

\def\bPsi{{\boldsymbol{\Psi}}}
\def\bsigma{{\boldsymbol{\sigma}}}
\def\btau{{\boldsymbol{\tau}}}
\def\bone{{\boldsymbol{1}}}

\def\qprob{{\mathbb{Q}}}
\def\tprob{\widetilde{\mathbb{P}}}
\def\prob{{\mathbb{P}}}
\def\E{{\mathbb{E}}}

\def\reals{{\mathbb{R}}}
\def\naturals{{\mathbb{N}}}
\def\ed{\stackrel{{\rm d}}{=}}
\def\normal{{\sf N}}
\def\ER{Erd\H{o}s-R\'enyi } 
\def\PSD{{\sf PSD}}
\def\SDP{{\sf SDP}}
\def\OPT{{\sf OPT}}
\def\Par{{\sf P}}
\def\xsdp{\boldsymbol{\hat{x}}^{\mbox{\tiny{SDP}}}}

\def\Greg{{\sf G}^{\mbox{\tiny {\sf reg}}}}

\def\sG{{\sf G}}

\def\Hnull{{\sf Hypothesis 0}}

\def\bbS{{\mathbb S}}
\def\rank{{\rm rank}}
\def\Tr{{\sf Tr}}
\def\pai{{\partial_i}}

\title {Semidefinite
  Programs on Sparse Random Graphs\\
and their Application to  Community Detection}

\author{
Andrea~Montanari\footnote{Department of Electrical
    Engineering and Department of Statistics, Stanford University,
    {\sf montanari@stanford.edu}}\;\;\;\;
and \;\;\;\;
Subhabrata~Sen\footnote{Department of Statistics, Stanford
  University, {\sf ssen90@stanford.edu}}}

\date{December 23, 2015}                                           % Activate to display a given date or no date

\begin{document}

\maketitle

\begin{abstract}
Denote by $\bA$ the adjacency matrix of an \ER graph with bounded
average degree.  We consider the problem of maximizing $\<\bA-\E\{\bA\},\bX\>$ over the set of
positive semidefinite matrices $\bX$ with diagonal entries $X_{ii}=1$.  We prove that for large (bounded) average degree $d$, the
value of this semidefinite program (SDP) is --with high probability-- $2n\sqrt{d} + n\, o(\sqrt{d})+o(n)$. For a random regular graph of
degree $d$, we prove that the SDP value is $2n\sqrt{d-1}+o(n)$,
matching a spectral upper bound. Informally, \ER graphs appear to behave similarly to random regular graphs for semidefinite programming.

We next consider the sparse, two-groups, symmetric community detection
problem (also known as planted partition). We establish that SDP
achieves the information-theoretically optimal detection threshold for
large (bounded) degree.
Namely, under this
model, the vertex set is partitioned into subsets of size $n/2$, with edge probability $a/n$ (within group) and $b/n$ (across). We prove
that SDP detects the partition with high probability provided
$(a-b)^2/(4d)> 1+o_{d}(1)$, with $d= (a+b)/2$. By comparison, the
information theoretic threshold for detecting the hidden partition is
$(a-b)^2/(4d)> 1$: SDP is nearly optimal for large bounded average degree.

Our proof is based on tools from different research areas:
$(i)$ A new `higher-rank' Grothendieck inequality for
symmetric matrices; $(ii)$ An interpolation method inspired from
statistical physics; $(iii)$ An analysis of the eigenvectors of
deformed Gaussian random matrices.
\end{abstract}

\thispagestyle{empty}

\newpage

\pagenumbering{arabic}

\section{Introduction and main results}

\subsection{Background}

Let $G = (V,E)$ be a random graph with vertex set $V=[n]$, and let $\bA_G\in
\{0,1\}^{n\times n}$ denote its adjacency matrix.
Spectral algorithms have proven extremely successful in analyzing 
the structure of such graphs under various probabilistic
models. Interesting tasks include  finding clusters, communities,
latent representations, collaborative filtering and so on \cite{alon1998finding,mcsherry2001spectral,ng2002spectral,coja2006spectral}. The underlying mathematical 
justification for these applications can be informally summarized as
follows (more precise statements are given below): 

\vspace{0.25cm}

\emph{If $G$ is dense enough, then $\bA_G-\E\{\bA_{G}\}$ is much
  smaller, in operator norm,  than $\E\{\bA_{G}\}$.}

\vspace{0.25cm}

(Recall that the operator norm of   a symmetric matrix $\bM$ is $\|\bM\|_{op}
    =\max(\xi_1(\bM),-\xi_n(\bM))$, with $\xi_{\ell}(\bM)$ the
    $\ell$-th largest eigenvalue of $\bM$.)

Random regular graphs provide the simplest model on which this intuition can be made precise  
Denoting by $\Greg(n,d)$ the uniform distribution over graphs with $n$
vertices and uniform degree $d$, we have, for $G\sim\Greg(n,d)$,
$\E\bA_G \approx (d/n)\bone\bone^{\sT}$, whence $\|\E\bA_G\|_2\approx
d$. On the other hand, the fact that random regular graphs are `almost
Ramanujan' \cite{Friedman} implies $\|\bA_G-\E\bA_G\|_{op}\le
2\sqrt{d-1}+o_n(1)\ll d$. Roughly speaking, the random part
$\bA_G-\E\bA_G$ is smaller than the expectation by a factor
$2/\sqrt{d}$.

The situation is not as clean-cut for random graph with irregular
degrees. To be definite, consider the \ER random graph distribution
$\sG(n,d/n)$ whereby each edge is present independently with
probability $d/n$ (and hence the average degree is roughly $d$). 
Also in this case $\E\bA_G \approx (d/n)\bone\bone^{\sT}$, whence $\|\E\bA_G\|_{op}\approx
d$. However, the largest eigenvalue of $\bA_G-\E A_G$ is of the order
of the square root of the maximum degree, namely $\sqrt{\log
  n/(\log\log n)}$ \cite{krivelevich2003largest}. Summarizing
\begin{align}
\|A_G-\E A_G\|_{op}= 
\begin{cases}
2\sqrt{d-1}\, (1+o(1)) & \mbox{ if $G\sim\Greg(n,d)$},\\
\sqrt{\log n/(\log\log n)} (1+o(1)) & \mbox{ if $G\sim\sG(n,d/n)$}.\\
\end{cases}\label{eq:MaxEigenvalue}
\end{align}
Further, for $G\sim\sG(n,d/n)$, the leading eigenvectors of
$\bA_G-\E\bA_G$ are concentrated near to high-degree vertices, and
carry virtually no information about the global structure of $G$. In
particular, they cannot be used for clustering.

Far from being a mathematical curiosity, this difference has far-reaching consequences: spectral algorithms are known fail, or to be
vastly suboptimal for random graphs with bounded average degree 
\cite{feige2005spectral,coja2010graph,keshavan2010matrix,decelle2011asymptotic,krzakala2013spectral}.
The community detection problem (a.k.a. `planted partition') is an example of this failure that
attracted significant attention recently. Let $\sG(n,a/n,b/n)$ be the
distribution over graph with $n$ vertices defined as follows. The
vertex set is partitioned uniformly at random into two subsets $S_1$, $S_2$ with
$|S_i|=n/2$. Conditional on this partition, edges are independent with
\begin{align}
\prob\big((i,j)\in E\big|S_1, S_2\big) = \begin{cases}
a/n & \mbox{ if $\{i,j\}\subseteq S_1$ or $\{i,j\}\subseteq S_2$,}\\
b/n & \mbox{ if $i\in S_1, j\in S_2$ or 
$i\in S_2, j\in S_1$.}
\end{cases}\label{eq:HiddenPart}
\end{align}
Given a single realization of such a graph, we would like to detect,
and identify the partition. Early work on this problem showed that simple spectral methods are
successful when $a=a(n)$, $b=b(n)\to\infty$ sufficiently fast. However
Eq.~(\ref{eq:MaxEigenvalue}) --and its analogue for the model
$\sG(n,a/n,b/n)$-- implies that this approach fails unless $(a-b)^2\ge
C \log n/\log\log n$. (Throughout $C$ indicates numerical constants.)

Several ideas have been developed to overcome this difficulty.
The simplest one is to simply remove from $G$ all vertices whose
degree is --say-- more than ten times larger than the average degree
$d$. Feige and Ofek \cite{feige2005spectral} showed that, if this
procedure is applied to $G\sim\sG(n,d/n)$, it yields a new graph  $G'$
that has roughly the same number of vertices as $G$, but
$\|\bA_G-\E\{\bA_G\}\|_{op}\le C\sqrt{d}$, with high probability.
The same trimming procedure was successfully applied in
\cite{keshavan2010matrix} to matrix completion, and in
\cite{coja2010graph,chin2015stochastic} to community detection. 
This approach has however several drawbacks. First, the specific
threshold for trimming is somewhat arbitrary and relies on the idea
that degrees should  concentrate around their average: this is not
necessarily true in actual applications. 
Second, it discards a subset of the data. Finally,  it is only optimal `up to
constants.'

A new set of spectral methods to overcome the same problem were
proposed and analyzed within the community detection problem
\cite{decelle2011asymptotic,krzakala2013spectral,mossel2013proof,massoulie2014community,bordenave2015non,le2015concentration}.
These methods construct a new matrix that replaces the adjacency matrix
$\bA_G$, and then compute its leading eigenvalues/eigenvectors. 
We refer to Section \ref{sec:Related} for further discussion.
These approaches are extremely interesting and mathematically
sophisticated. In particular, some of them have been proved to have an optimal
detection threshold under the model $\sG(n,a/n,b/n)$ \cite{mossel2013proof,massoulie2014community,bordenave2015non}. Unfortunately
 they rely on delicate properties of the underlying
probabilistic model. For instance,  they are not
robust to an adversarial addition of $o(n)$ edges (see Section \ref{sec:Generalization}). 

\subsection{Main results (I): \ER and regular random graphs}

Semidefinite programming (SDP) relaxations provide a different
approach towards overcoming the limitations of spectral algorithms.
We denote the cone of $n\times n$ symmetric positive semidefinite
matrice by $\PSD(n) \equiv\{\bX\in\reals^{n\times n}:\; \bX\succeq
0\}$. The convex set of positive-semidefinite matrices with diagonal
entries equal to one is denoted by  
\begin{align}
\PSD_1(n) \equiv\big\{\bX\in\reals^{n\times n}:\; \bX\succeq
0, \;X_{ii}=1\forall i\in [n]\big\}\, .
\end{align}
The set $\PSD_1(n)$  is also known as the \emph{elliptope}. Given a
matrix $\bM$, we define\footnote{Here and below
  $\<\bA,\bB\>=\Tr(\bA^{\sT}\bB)$ is the usual scalar product between matrices.}
\begin{align}
\SDP(\bM) \equiv \max\big\{ \<\bM,\bX\>\, :\;\;
\bX\in\PSD_1(n)\big\}\, .  \label{eq:SDP.DEF}
\end{align}
It is well known that approximate  information about the extremal cuts
of $G$ can be obtained by computing $\SDP(\bA_G)$
\cite{goemans1995improved}. 

The main result of this paper is that the above SDP 
is also nearly optimal in extracting information about sparse random
graphs. In particular, it eliminates the irregularities due to
high-degree vertices, cf. Eq.~(\ref{eq:MaxEigenvalue}). Our first
result characterizes the value of $\SDP(\bA_G-\E\{\bA_G\})$ for $G$ an
\ER random graph with large bounded degree\footnote{Throughout the
  paper, $O(\, \cdot\, )$, $o(\,\cdot\,)$, and $\Theta(\,\cdot\,)$
  refer to the usual $n \to \infty$ asymptotic, while $O_{d}(\,\cdot\,)$, $o_{d}(\,\cdot\,)$ 
and $\Theta_{d}(\,\cdot\,)$ are used
to describe the $d \to \infty$ asymptotic regime. We say that a sequence of events $B_n$ occurs with high probability (w.h.p.) if $\prob(B_n) \to 1$ 
as $n\to \infty$. Finally, for random $\{X_n\}$ and 
non-random $f: \reals_{>0} \to \reals_{>0}$, we say 
that $X_n = o_{d}(f(d))$ w.h.p. as $n\to \infty$ if 
there exists non-random $g(d) = o_{d}(f(d))$ such 
that the sequence $B_n = \{ |X_n| \leq g(d)\}$ occurs w.h.p. 
(as $n\to \infty$).}. (Its proof is given in Appendix \ref{sec:ProofMain}.)
\begin{theorem}\label{thm:Main}
Let $G\sim \sG(n,d/n)$ be an \ER random graph with edge probability
$d/n$,  $\bA_G$ its adjacency matrix, and $\bAc_G \equiv
\bA_G-\E\{\bA_G\}$ its centered adjacency matrix. 
Then there exists $C=C(d)$ such that with probability at least $1-C\,
e^{-n/C}$, we have
\begin{align}
\frac{1}{n}\SDP(\bAc_G) = 2 \sqrt{d} + o_{d}(\sqrt{d})\, ,\label{eq:MaxLimit}\;\;\;\;\;
\frac{1}{n}\SDP(-\bAc_G)  = 2 \sqrt{d} + o_{d}(\sqrt{d})\, .
\end{align}
\end{theorem}
Note that $\SDP(\bAc_G)\le n\xi_1(\bAc_G)$ (here and in the following
$\xi_1(\bM)\ge \xi_2(\bM)\ge\dots\xi_n(\bM)$ denote the
eigenvalues of the symmetric matrix $\bM$). However, while
$\xi_1(\bAc_G)$ is sensitive to vertices of atypically large
degree, cf. Eq.~(\ref{eq:MaxEigenvalue}), $\SDP(\bAc_G)$ appears to be
sensitive only to the average degree. Intuitively, the constraint
$X_{ii}=1$ rules out the highly localized eigenvectors that are
responsible for $\xi_1(\bAc_G) \approx\sqrt{\log n/\log\log n}$.

Another way of interpreting Theorem \ref{thm:Main} is that 
\ER random graphs behave, with respect to SDP as random regular graphs
with the same average degree. Indeed, we have the following more
precise result for regular graphs. (See Appendix \ref{app:Regular} for the proof.)
\begin{theorem}\label{thm:Regular}
Let $G\sim \Greg(n,d)$ be a random regular graph with degree
$d$, and $\bAc_G \equiv \bA_G-\E\{\bA_G\}$ its centered adjacency
matrix. 
Then, with high probability
\begin{align}
\frac{1}{n}\SDP(\bAc_G) = 2 \sqrt{d-1} + o_n(1)\, ,\;\;\;\;\;
\frac{1}{n}\SDP(-\bAc_G)  = 2 \sqrt{d-1} + o_n(1)\, .
\end{align}
\end{theorem}

\begin{remark}
The quantity $\SDP(\bAc_G)$ can also be thought as a relaxation
of the problem of maximizing 
$\sum_{i,j=1}^nA_{ij}\sigma_i\sigma_j$ over $\sigma_i\in\{+1,-1\}$,
$\sum_{i=1}^n\sigma_i=0$. The result of our  companion
paper \cite{dembo2015extremal} implies that this has --with high
probability--
value $2n \Par_*\sqrt{d}+n \, o_d(\sqrt{d})$ (see
\cite{dembo2015extremal} for a definition of $\Par_*$). We
deduce that --with high probability-- the SDP relaxation overestimates the optimum  by a factor
$1/\Par_*+o_{d}(1)$ (where $1/\Par_*\approx 1.310$). 
\end{remark}

\begin{remark}
For the sake of simplicity, we stated Eq.~(\ref{eq:MaxLimit}) in
  asymptotic form. However, our proof provides  quantitative bounds
  on the error terms. In particular, the $o_{d}(\sqrt{d})$
  term is upper bounded by $C d^{2/5}\log(d)$, for $C$ a
  numerical constant.
\end{remark}

\subsection{Main results (II): Hidden partition problem}
\label{sec:MainPartition}

We next apply the SDP defined in Eq.~(\ref{eq:SDP.DEF}) to the
community detection problem. 
To be definite we will formalize this as a binary hypothesis
testing problem, whereby we want to determine --with high probability
of success-- whether the random graph under consideration has a
community structure or not. The estimation version of the problem,
i.e. the question of determining --approximately-- a
partition into communities, can be addressed by similar techniques.

We are given a \emph{single} graph $G=(V,E)$ over
$n$ vertices and we have to decide which of the following holds:
\begin{description}
\item[{\sf Hypothesis 0:}] $G\sim \sG(n,d/n)$ is an \ER random graph with edge
  probability $d/n$, $d=(a+b)/2$. We denote the
  corresponding distribution over graphs by $\prob_0$.
\item[{\sf Hypothesis 1:}]  $G\sim \sG(n,a/n,b/n)$ is an random graph
  with a planted partition and edge probabilities $a/n$, $b/n$. We denote the
  corresponding distribution over graphs by $\prob_1$.
\end{description}
A statistical test takes as input a graph $G$, and returns
$T(G)\in\{0,1\}$ depending on which hypothesis is estimated to hold.
We say that it is successful with high probability if
$\prob_0(T(G)=1)+\prob_1(T(G)=0)\to 0$ as $n\to\infty$.

Theorem \ref{thm:Main} indicates that, under \Hnull, 
we have $\SDP(\bA_G-(d/n)\bone\bone^{\sT})= 2n\sqrt{d} +
n\,o_{d}(\sqrt{d})$.
This suggests the following test:
\begin{align}
T(G;\delta) = \begin{cases}
1 & \mbox{ if $\SDP(\bA_G-(d/n)\bone\bone^{\sT})\ge 2n(1+\delta)\sqrt{d}$,}\\
0 & \mbox{ otherwise.}\\
\end{cases}\label{eq:TestDef}
\end{align}
Mossel, Neeman, Sly \cite{mossel2012stochastic} proved that no test can be successful with
high probability if $(a-b)<\sqrt{2(a+b)}$. Polynomially computable
tests that achieve this threshold were developed in
\cite{mossel2013proof,massoulie2014community,bordenave2015non} using
advanced spectral methods. As mentioned, these approaches  can be
fragile to perturbations of the precise probabilistic model, cf. Section \ref{sec:Generalization}. 

Our next result addresses the fundamental question: \emph{Does the
  SDP-based test achieve the information theoretic threshold?} Notice
that the recent work of \cite{guedon2014community} falls short of
answering this question since it requires the vastly sub-optimal
condition $(a-b)^2\ge 10^4(a+b)$. (We refrer to Appendix \ref{sec:ProofMain} for
its proof.)
\begin{theorem}\label{thm:SDP_Test}
Assume, for some $\eps>0$,
\begin{align}
\frac{a-b}{\sqrt{2(a+b)}} \ge 1+\eps\, .\label{eq:ConditionFactor1}
\end{align}
Then there exists $\delta_*=\delta_*(\eps)>0$ and $d_* = d_*(\eps)>0$
such that the following holds. If $d=(a+b)/2\ge d_*$, then the
SDP-based test $T(\,\cdot\,;\delta_*)$ succeeds 
with high probability. 

Further, the error probability is at most
$Ce^{-n/C}$ for $C=C(a,b)$ a constant.
\end{theorem}
\begin{remark}
This theorem guarantees that 
SDP is nearly optimal for large but bounded degree $d$. 
By comparison, the naive spectral test that returns $T_{\rm spec}(G) = 1$ if
$\lambda_1(\bA_G)\ge \theta_*$  and $T_{\rm spec}(G) = 0$ otherwise
(for any threshold value $\theta_*$) is sub-optimal by an unbounded
factor for  $d=O(1)$.
\end{remark}

\begin{remark}
One might wonder why we consider large degree asymptotics $d=(a+b)/2\to\infty$
instead of trying to establish a threshold at $(a-b)/\sqrt{2(a+b)}=1$ for
fixed $a$, $b$. Preliminary non-rigorous calculation
\cite{OursReplicas} suggest that indeed this is necessary. For fixed
$(a+b)$ the SDP threshold does not coincide with the optimal one.
\end{remark}

\begin{remark}
For the sake of simplicity, we formulated the community detection
problem as an hypothesis testing problem. A related (somewhat more
challenging) task is to estimate the hidden partition better than by
random guessing. In Section \ref{sec:Estimation} we will show that, under the same conditions of
Theorem \ref{thm:SDP_Test}, we can assign vertices making at most $(1-\Delta)n/2$ mistakes 
(with high probability for some $\Delta$ bounded away from $0$).
\end{remark}

We will discuss related work in the next section, then provide an
outline of the proof ideas in Section \ref{sec:Strategy}, and finally
discuss extension of the above results in Section
\ref{sec:Generalization}.
Detailed proofs are deferred to the appendix.

\subsection{Notations}

Given $n\in\naturals$, we let $[n] = \{1,2,\dots,n\}$
denote the set of first $n$ integers.  We write $|S|$ for the cardinality of a
set $S$.  We will use
lowercase boldface (e.g. $\bv = (v_1,\dots,v_n)$, $\bx = (x_1,\dots,x_n)$, etc.) 
for  vectors and uppercase boldface (e.g. $\bA = (A_{i,j})_{i,j\in[n]}$, $\bY= (Y_{i,j})_{i,j\in[n]}$, etc.)
for matrices.
Given a symmetric matrix $\bM$, we let $\xi_1(\bM)\ge
\xi_2(\bM)\ge \dots\ge \xi_n(\bM)$ be its ordered eigenvalues (with
$\xi_{\max}(\bM) = \xi_1(\bM)$, $\xi_{\min}(\bM) = \xi_n(\bM)$).
In particular $\bone_n = (1,1,\dots, 1)\in\reals^n$ is the all-ones vector,
$\id_{n}$ the identity matrix, and $\bfe_i\in \reals^{n}$ is the $i$'th standard unit vector.

For $\bv\in\reals^m$, $\|\bv\|_p =
(\sum_{i=1}^p|v_i|^p)^{1/p}$ denotes its $\ell_p$ norm (extendend in
the standard way to $p=\infty$). For a matrix $\bM$, we denote by
$\|\bM\|_{p\to q} = \sup_{\bv\neq 0}\|\bM\bv\|_q/\|\bv\|_q$ its
$\ell_p$-to-$\ell_q$ operator norm, with the standard shorthands
$\|\bM\|_{op} \equiv  \|\bM\|_{2} \equiv \|\bM\|_{2\to 2}$.

Throughout \emph{with high probability} means `with probability
converging to one as $n\to\infty$.'  We follow the standard Big-Oh
notation for asymptotics. We will be interested in bounding error
terms with respect to $n$ and $d$. Whenever not clear from the
contest, we indicate in subscript the variable that is large. For
instance
$f(n,d) = o_d(1)$ means that there exists a function $g(d)\ge 0$
independent of $n$ such that $\lim_{d\to\infty}g(d) = 0$ and
$|f(n,d)|\le g(d)$. (Hence $f(n,d) = \cos(0.1 n)/d =o_d(1)$ but
$f(n,d) = \log( n)/d \neq o_d(1)$.)

A random graph has a law (distribution), which is a probability
distribution over graphs with the same vertex set $V=[n]$. Since we
are interested in the $n\to\infty$ asymptotics, it will be implicitly
understood that one such distribution is specified for each $n$.

We will use $C$ (or $C_0$, $C_1$,\dots) to denote constants, that will
change from point to point. Unless otherwise stated, these are
universal constants. 
%
%***************************************************************
%
\section{Further related literature}
\label{sec:Related}

Few results have been proved about the behavior of classical SDP relaxations
on sparse random graphs and --to the best of our knowledge-- none of
these earlier results is tight.

Significant amount of work has been devoted to analyzing SDP
hierarchies on random CSP instances
\cite{grigoriev2001linear,schoenebeck2008linear}, and --more
recently-- on (semi-)random Unique games instances
\cite{kolla2011play}. These papers typically prove only one-side
bounds that are not claimed to be sharp as the number of variables diverge.

Coja-Oghlan \cite{coja2003lovasz} studies the
value of  Lov\'asz  theta function $\vartheta(G)$, for $G\sim\sG(n,p)$ a
\emph{dense} \ER random graph,
 estabilishing $C_1\sqrt{n/p}\le \vartheta(G)\le
C_2\sqrt{n/p}$ with high probability. As in the previous cases, this result is not tight.

Ambainis et
al. \cite{ambainis2012quantum}  study an SDP similar to
(\ref{eq:SDP.DEF}), for $\bM$ a \emph{dense} random matrix with
i.i.d. entries. One of their main results is analogous to a special
case of our  Theorem \ref{thm:Gaussian}.$(b)$ below --namely, to the case $\lambda=0$.
(We prefer to give an independent --simpler-- proof also of this case.)

Several papers have been devoted to SDP approaches for community
detection and the related `synchronization' problem.
A partial list includes
\cite{bandeira2014multireference,abbe2014exact,hajek2014achieving,hajek2015achieving,awasthi2015relax}.
These papers focus on finding sufficient conditions under which the
SDP recovers \emph{exactly} the unknown signal.
For instance, in the context of the hidden partition model
(\ref{eq:HiddenPart}), this requires  diverging degrees $a,b=\Theta(\log n)$ 
\cite{abbe2014exact,hajek2014achieving,hajek2015achieving}. 
SDP was proved in \cite{hajek2014achieving} to achieve the information-theoretically optimal
threshold for exact reconstruction.
The techniques to prove this type of result are very different from
the ones employed here: since the (conjectured) optimum is known
explicitly, it is sufficient to certify it through a dual witness.

The only  result on community detection that compares to ours was recently proven by Guedon and
Vershynin \cite{guedon2014community}.
Their work uses the classical Grothendieck inequality to
establish upper bounds on the estimation error of
SDP. The resulting bound applies only under the condition $(a-b)^2\ge 10^4
(a+b)$. This condition is vastly sub-optimal with respect to the
information-theoretic threshold  $(a-b)^2> 2
(a+b)$ established in
\cite{mossel2012stochastic,mossel2013proof,massoulie2014community}
(and is unlikely to be satisfied by realistic graphs). In particular,
the results of  \cite{guedon2014community} leave open the central question:
is SDP to be discarded in favor of the spectral methods of
\cite{mossel2013proof,massoulie2014community},
or is the sub-optimality just an outcome of the analysis?

In this paper we provide evidence indicating that SDP is in fact nearly optimal
for community detection. While we also make use of a Grothendieck
inequality as in \cite{guedon2014community}, this is only one step
(and not the most challenging) in a  significantly longer
argument. Let us emphasize that the gap between the ideal threshold
at $(a-b)/\sqrt{2(a+b)} =1$,
and the guarantees of \cite{guedon2014community} cannot be
filled simply by carrying out more carefully the same proof strategy.
In order fill the gap we need
to develop several new ideas: $(i)$ A new (higher rank) Grothendieck
inequality; $(ii)$ A smoothing  of the original graph parameter
$\SDP(\,\cdot\,)$; $(iii)$ An interpolation argument; $(iv)$ A sharp
analysis of SDP for Gaussian random matrices.
%
%***************************************************************
%
\section{Proof strategy}
\label{sec:Strategy}

Throughout, we denote by $\bAc_G=\bA_G-(d/n)\bone\bone^{\sT}$ the
centered adjacency matrix of $G\sim\sG(n,d/n)$ or $G\sim\sG(n,a/n,b/n)$. 
Our proofs of Theorem \ref{thm:Main} and Theorem \ref{thm:SDP_Test}
follows a similar strategy that can be summarized as follows:
\begin{description}
\item[Step 1: Smooth.] We replace the function $\bM\mapsto\SDP(\bM)$,
  by a smooth function $\bM\mapsto \Phi(\beta,k;\bM)$ that depends on two
  additional parameters $\beta\in\reals_{\ge 0}$ and
  $k\in\naturals$. We prove that, for $\beta, k$ large (and $\bM$
  sufficiently `regular'), $|\SDP(\bM)-\Phi(\beta,k;\bM)|$ can be made
  arbitrarily small,
  uniformly in the matrix dimensions. This in particular requires
  developing a new (higher rank) Grothendieck-type inequality, which
  is of independent interest, see Section \ref{sec:Gro}.
\item[Step 2: Interpolate.] We use an interpolation method (analogous
  to the Lindeberg method) to compare the value $\Phi(\beta,k;\bAc_G)$
  to $\Phi(\beta,k;\bB)$, where $\bB\in\reals^{n\times n}$ is a symmetric
  Gaussian matrix with independent entries. More precisely, we use
  $B_{ij}\sim \normal(0,1/n)$ to approximate $G\sim\sG(n,d/n)$ and
  $B_{ij}\sim \normal(\lambda/n,1/n)$ to approximate  the hidden
  partition model $G\sim\sG(n,a/n,b/n)$, with $\lambda \equiv
  (a-b)/\sqrt{2(a+b)}$. Further detail is provided in Section \ref{sec:Interpolation}.

 Note that the interpolation/Lindeberg method requires $\bM\mapsto
 \Phi(\beta,k;\bM)$ to be differentiable, which is the reason for Step
 1 above.
\item[Step 3: Analyze.] We finally carry out an analysis of $\SDP(\bB)$
  with $\bB$ distributed according to the above Gaussian models. In
  doing this we can take advantage of the high degree of symmetry of
  Gaussian random matrices. This part of the proof is relatively
  simple for Theorem \ref{thm:Main}, but becomes challenging in the
  case of Theorem \ref{thm:SDP_Test}, see Section \ref{sec:Gaussian}.
\end{description}
(The proof of Theorem \ref{thm:Regular} is more direct and will be
presented in Appendix \ref{app:Regular}). In the next subsections we will provide
further details about each of these steps.  The formal proofs
of Theorem \ref{thm:Main} and Theorem \ref{thm:SDP_Test} are presented
in Appendix \ref{sec:ProofMain}, with technical lemmas in other appendices..

The construction of the smooth function $\Phi(\beta,k;\bM)$ is
inspired from statistical mechanics. As an intermediate step, define
the following rank-constrained version of the SDP (\ref{eq:SDP.DEF})
\begin{align}
\OPT_k(\bM) &\equiv \max\big\{ \<\bM,\bX\>\, :\;\;
\bX\in\PSD_1(n)\, ,\;\; \rank(\bX)\le k\big\} \label{eq:OPT.DEF}\\
& = \max\big\{ \sum_{i,j=1}^nM_{ij}\<\bsigma_i,\bsigma_j\>\, :\;\;
\bsigma_i\in \bbS^{k-1}\big\}\, ,
\end{align}
where  $\bbS^{k-1} = \{ \bsigma \in\reals^k:\; \|\bsigma\|_2 = 1\}$ be the unit
sphere in $k$ dimensions. We then define $\Phi(\beta,k;\bM)$ as the
following  log-partition function
\begin{align}
\Phi(\beta,k;\bM) &\equiv \frac{1}{\beta}\, \log\left\{\int \,
                    \exp\Big\{\beta
                    \sum_{i,j=1}^{n}M_{ij}\<\bsigma_i,\bsigma_j\>\Big\}\,\de\nu(\bsigma)\right\}\,
                    .
\end{align}
Here $\bsigma=(\bsigma_1,\bsigma_2,\dots,\bsigma_n)\in (\bbS^{k-1})^n$
and we denote by $\de\nu(\,\cdot\,)$ the uniform measure on
$(\bbS^{k-1})^n$  (normalized to $1$, i.e. $\int \de\nu(\bsigma) =
1$).

It is easy to see that $\lim_{\beta\to\infty}\Phi(\beta,k;\bM) =
\OPT_k(\bM)$, and $\OPT_n(\bM) = \SDP(\bM)$. For carrying out the
above proof strategy we need to bound the errors $|\Phi(\beta,k;\bM) -
\OPT_k(\bM)|$ and $|\OPT_k(\bM) -\SDP(\bM)|$ uniformly in $n$. 

\subsection{Higher-rank Grothendieck inequalities and
  zero-temperature limit}
\label{sec:Gro}

In order to bound the error $|\OPT_k(\bM) -\SDP(\bM)|$ we develop a
new Grothendieck-type inequality  which is of independent interest.
\begin{theorem}\label{thm:Gro}
For $k\ge 1$, let $\bg\sim\normal(0,\id_{k}/k)$ be a vector with
i.i.d. centered normal entries with variance $1/k$, and define
$\alpha_k \equiv (\E\|\bg\|_2)^2$.

Then, for any symmetric matrix $\bM\in\reals^{n\times n}$, we have the inequalities
\begin{align}
\SDP(\bM) \ge \OPT_k(\bM) &\ge\alpha_k \SDP(\bM) - (1-\alpha_k)\,
  \SDP(-\bM) \, ,\label{eq:Gro}\\
 \OPT_k(\bM) &\ge \big(2-\alpha_k^{-1}\big)\SDP(\bM) - \big(\alpha_k^{-1}-1\big)\,
  \OPT_k(-\bM) \label{eq:GroBis}\, . 
\end{align}
\end{theorem}

\begin{remark}
The upper bound in Eq.~(\ref{eq:Gro}) is trivial.
Further, it follows from Cauchy-Schwartz that $\alpha_k\in (0,1)$ for all
$k$. Also  $\|\bg\|^2_2$ is a chi-squared random variable with $k$
degrees of freedom and hence
\begin{align}
\alpha_k = \frac{2\Gamma((k+1)/2)^2}{k\Gamma(k/2)^2} = 1-\frac{1}{2k}
  +O(1/k^2)\, .\label{eq:AlphaK}
\end{align}
Substituting in Eq.~(\ref{eq:Gro}) we get, for all $k\ge k_0$  with
$k_0$ a
sufficiently large constant, and assuming $\SDP(\bM)>0$,
\begin{align}
\Big(1-\frac{1}{k}\Big)\SDP(\bM) -\frac{1}{k}\, |\SDP(-\bM)|
  \le \OPT_k(\bM) \le \SDP(\bM)\, .\label{eq:GroSimple}
\end{align}
In particular, if $|\SDP(-\bM)|$ is of the same order as $\SDP(\bM)$, we
conclude that $\OPT_k(\bM)$ approximates $\SDP(\bM)$ with a relative error of
order $O(1/k)$.
\end{remark}

The classical Grothendieck inequality concerns non-symmetric bilinear
forms \cite{grothendieck1996resume}.
A  Grothendieck inequality for symmetric matrices was established in
\cite{nemirovski1999maximization,megretski2001relaxations}
(see also \cite{alon2006quadratic} for generalizations)
and states that, for a constant $C$, 
\begin{align}
\OPT_1(\bM) \ge \frac{1}{C\log n}\,  \SDP(\bM)\, .
\end{align}
Higher-rank  Grothendieck inequalities were developed in the setting
of general graphs in
\cite{briet2010grothendieck,briet2010positive}. However, 
constant-factor approximations were not established for the present
problem (which corresponds to the the complete graph case in
\cite{briet2010grothendieck}).

Constant factor approximations exist for $\bM$ positive semidefinite 
\cite{briet2010positive}. 
We note that Theorem \ref{thm:Gro} implies the inequality of \cite{briet2010positive}. 
Using $\SDP(-\bM)\le -\xi_{\rm min}(\bM)$ in Eq.~(\ref{eq:Gro}),
we obtain the inequality of \cite{briet2010positive} for the positive semidefinite
matrix $\bM-\xi_{\rm min}(\bM)\id$. 
On the other hand, the result of \cite{briet2010positive} is too weak
for our applications.   We want to apply Theorem
\ref{thm:Gro} --among others-- to $\bM=  \bAc_G$ with $\bAc_G$ the
adjacency matrix of $G\sim\sG(n,d/n)$.
This matrix is non-positive definite, and in a dramatic
way with smallest eigenvalue satisfying $-\xi_{\rm
  min}(\bAc_G) \approx (\log n/(\log\log
n))^{1/2}\gg \SDP(-\bAc_G)$). 

In summary, we could not use the vast literature on Grothendieck-type
inequality to prove our main result, Theorem \ref{thm:Main}, which
motivated us to develop Theorem \ref{thm:Gro}.

Theorem \ref{thm:Gro} will allow to bound $|\SDP(\bM)-\OPT_k(\bM)|$ 
for $\bM$ either a centered adjacency matrix or a Gaussian matrix. The
next lemma bounds the `smoothing error' 
$|\Phi(\beta,k;\bM)-\OPT_k(\bM)|$.
\begin{lemma}\label{lemma:ZeroTemperature}
There exists an absolute constant $C$ such that for  any $\eps\in
(0,1]$ the following holds.
If  $\|\bM\|_{\infty\to 2} \equiv
\max\{\|\bM\bx\|_{2}:\;\; \|\bx\|_{\infty}\le 1\}\le L\sqrt{n}$, then
\begin{align}
\Big|\frac{1}{n}\Phi(\beta,k;\bM)-\frac{1}{n}\OPT_k(\bM)\Big|\le
  2L\eps\sqrt{k} +\frac{k}{\beta}\log\frac{C}{\eps}\, .\label{eq:TemperatureBound}
\end{align}
\end{lemma}

\subsection{Interpolation}
\label{sec:Interpolation}

Our next step consists in comparing the adjacency matrix of random graph 
$G$ with a suitable Gaussian random matrix, and bound the error in the corresponding 
log-partition function $\Phi(\beta,k;\,\cdot\,)$. 

Let us recall the definition of Gaussian orthogonal ensemble
$\GOE(n)$. We have $\bW\sim \GOE(n)$ if $\bW\in\reals^{n\times n}$ is symmetric with
 $\{W_{i,j}\}_{1\le i\le j\le n}$ independent, with distribution $W_{ii}\sim\normal(0,2/n)$ and $W_{ij}\sim\normal(0,1/n)$ for $i<j$.
We then define, for $\lambda\ge 0$, the following \emph{deformed $\GOE$} matrix:
\begin{align}
\bB(\lambda) \equiv \frac{\lambda}{n}\, \bone\bone^{\sT}+ \bW\, ,\label{eq:Bdefinition}
\end{align}
where $\bW\sim\GOE(n)$. The argument $\lambda$ will be omitted if clear from the context.
The next lemma establishes the necessary comparison bound.
Note that we state it for $G\sim\sG(n,a/b,b/n)$ a random graph from the hidden partition model,
but it obviously applies to standard \ER random graphs by setting $a=b=d$.
\begin{lemma}\label{lemma:Interpolation}
Let $\bAc_G= \bA_G-(d/n)\bone\bone^{\sT}$ be the centered adjacency matrix of $G\sim\sG(n,a/n,b/n)$, whereby
$d= (a+b)/2$. Define $\lambda=(a-b)/2\sqrt{d}$. Then there exists an
absolute constant $n_0$
such that, if $n\ge \max(n_0,(15d)^2)$,
\begin{align}
\left|\frac{1}{n}\E\Phi\big(\beta,k;\bAc_G/\sqrt{d}\big)-\frac{1}{n}\E\Phi(\big(\beta,k;\bB(\lambda)\big)\right|\le
\frac{2\beta^2}{\sqrt{d}} +\frac{8\lambda^{1/2}}{d^{1/4}}\, .
\end{align}
\end{lemma}
Note that this lemma bounds the difference in expectation. We will use
concentration of measure to transfer this result to a bound holding
with high probability.

Interpolation (or `smart path') methods have a long history in
probability theory, dating back to Lindeberg's beautiful proof of the
central limit theorem \cite{lindeberg1922neue}.
Since our smoothing construction yields a log-partition function $\Phi(\beta,k;\bM)$,
our calculations are similar to certain proofs in statistical mechanics.
A short list of statistical-mechanics inspired results in probabilistic combinatorics includes
\cite{FranzLeone,FranzLeoneToninelli,BGT,panchenko2004bounds,GuerraToninelliDiluted}.
In our companion paper \cite{dembo2015extremal}, we used a similar approach
to characterize the limit value of the minimum bisection of \ER and random regular graphs.

\subsection{SDPs for Gaussian random matrices}
\label{sec:Gaussian}

The last part of our proof analyzes the Gaussian model
(\ref{eq:Bdefinition}). This type of random matrices have attracted
a significant amount of work within statistics (under the name of
`spiked model') and probability theory (as `deformed
Wigner --or GOE-- matrices'), aimed at characterizing their eigenvalues
and eigenvectors.
A very incomplete list of references includes
\cite{baik2005phase,feral2007largest,capitaine2011free,benaych2012large,bloemendal2013limits,pizzo2013finite,knowles2013isotropic}. 
A key phenomenon unveiled by these works is the so-called
\emph{Baik-Ben Arous-Pech\'e (or BBAP) phase transition}. In its
simplest form (and applied to the matrix of
Eq.~(\ref{eq:Bdefinition})) this predicts a phase transition in the
largest eigenvalue of $\bB(\lambda)$  
\begin{align}
\lim_{n\to\infty}\xi_1(\bB(\lambda)) = 
\begin{cases}
2 & \mbox{ if $\lambda\le 1$,}\\
\lambda+\lambda^{-1} & \mbox{ if $\lambda> 1$.}
\end{cases}\label{eq:BBAP}
\end{align}
(This limit can be interpreted as holding in probability.)
Here, we establish an analogue of this result for the SDP value.
\begin{theorem}[SDP phase transition for deformed GOE matrices]\label{thm:Gaussian}
Let $\bB=\bB(\lambda)\in\reals^{n\times n}$ be a symmetric matrix distributed
according to the model (\ref{eq:Bdefinition}). Namely $\bB=\bB^{\sT}$
with $\{B_{ij}\}_{i\le j}$ independent random variables, where
$B_{ij}\sim \normal(\lambda/n,1/n)$ for $1\le i<j\le n$
and $B_{ii}\sim \normal(\lambda/n,2/n)$ for $1\le i\le n$. Then
\begin{enumerate}
\item[$(a)$] If $\lambda\in [0,1]$, then for any $\eps>0$, we have
  $\SDP(\bB(\lambda))/n\in [2-\eps,2+\eps]$ with probability
  converging to one as $n\to\infty$.
\item[$(b)$] If $\lambda>1$, then there exists $\Delta(\lambda)>0$ such that
$\SDP(\bB(\lambda))/n\ge 2+\Delta(\lambda)$ with probability
  converging to one as $n\to\infty$.
\end{enumerate}
\end{theorem}

As mentioned above, we obviously have $\SDP(\bB)/n\le \xi_1(\bB)$. The
first part of this theorem (in conjunction with Eq.~(\ref{eq:BBAP}))
establishes that the upper bound is essentially tight of $\lambda\le
1$. On the other hand, we expect the eigenvalue upper bound not to be tight for
$\lambda>1$ \cite{OursReplicas}. Nevertheless, the second part of
our theorem establishes a phase transition taking place at
$\lambda=1$ as for the leading eigenvalue.

\begin{remark}
The phase transition in the leading eigenvalue has a high degree 
of universality. In particular, Eq.~(\ref{eq:BBAP}) remains correct if
the model (\ref{eq:Bdefinition}) is replaced by $\bB' =
\lambda\bv\bv^{\sT}+\bW$, with $\bv$ an arbitrary unit vector. 
On the other hand, we expect the phase transition in $\SDP(\bB')/n$ to
depend --in general-- on the vector $\bv$, and in particular on how
`spiky' this is.
\end{remark}
%
%***************************************************************
%

\section{Other results and generalizations}
\label{sec:Generalization}

While our was focused on a relatively simple model, the techniques
presented here allow for several generalizations. We discuss them
briefly here.

\subsection{Estimation}
\label{sec:Estimation}

For the sake of simplicity, we formulated
community detection as an \emph{hypothesis testing} problem. It is
interesting to  consider the associated \emph{estimation} problem,
that requires to estimate the hidden partition $V =S_1\cup S_2$. 

We encode the ground truth using the vector $\bxz\in\{+1,-1\}^n$, with
$x_{0,i}=+1$ if $i\in S_1$, and $x_{0,i} =-1$ if $i\in S_2$. An
estimator is a map\footnote{Earlier work sometimes assumes $\hbx:\cG_n\to \{+1,-1\}^n$, i.e. forbids the estimate $0$.
For our purposes, the two formulations are equivalent: we can always `simulate' $\hat{x}_i=0$ by letting $\hat{x}_i\in\{+1,-1\}$ uniformly at random.} 
$\hbx:\cG_n\to \{+1,0,-1\}^n$ with $\cG_n$ the space
of graphs over $n$ vertices. It is proved in
\cite{mossel2012stochastic} that no estimator is substantially better
than random guessing for $G\sim\sG(n,a/n,b/n)$, with
$\lambda=(a-b)/\sqrt{2(a+b)}<1$. More precisely, for $\lambda<1$, any
estimator achieves vanishing correlation with the ground truth:
$|\<\hbx(G),\bxz\>|=o(n)$ with high probability.

We construct a randomized SDP-based estimator $\xsdp(G)$ as follows
(we will denote expectation and probability with respect to tha
algorithm's randomness by $\Ealg(\,\cdot\,)$ and $\Palg(\,\cdot\,)$):
\begin{itemize}
\item[$(i)$] Partition the edge set $E=E_1\cup E_2$ by letting
  $(i,j)\in E_2$ independently for  each edge $(i,j)\in E$, with
  probability $\Palg\big((i,j)\in E_2\big)=\delta_n/(1+\delta_n)$,  $\delta_n= n^{-1/2}$, and
  $(i,j)\in E_1$ otherwise. Denote by $G_1=(V,E_1)$, and $G_2=(V,E_2)$
  the resulting graphs.
\item[$(ii)$] Compute an optimizer $\bX_*$ of the SDP
  (\ref{eq:SDP.DEF}), $\bM=\bAc_{G_1}$ (i.e. a matrix $\bX_*\in\PSD_1(n)$ such
  that $\<\bAc_{G_1},\bX_*\> = \SDP(\bAc_{G_1})$). 
\item[$(iii)$] Compute the eigenvalue decomposition $\bX_*=
  \sum_{i=1}^n\xi_i\bv_i\bv_i^{\sT}$, and let $\bv_i =
  (v_{i,1},v_{i,2},\dots,v_{i,n})$ denote the $i$-th eigenvector. For each $i,j\in [n]$ define $\hbx^{(i,j)}\in\{+1,0,-1\}$ 
by $\hx^{(i,j)}_\ell = \sign(\bv_{i})_{\ell}$ if $|v_{i,\ell}|\ge $
  = $|v_{i,j}|$ and $\hx^{(i,j)}_\ell =  0$ otherwise. (In words, $\hbx^{(i,j)}$ is obtained from $\bv_i$ by zeroing entries
with magnituude below $|v_{i,j}|$ and taking the sign of those above).
\item[$(iv)$]  Select $(I,J) = \arg\max_{i,j\in [n]}\<\hbx^{(i,j)},\bA_{G_2}\hbx^{(i,j)}\>$, and return $\xsdp(G) = \hbx^{(I,J)}$.
\end{itemize}
The next results implies that --for large bounded average degree $d$--
this estimator has a nearly optimal threshold.
\begin{theorem}\label{thm:Estimation}
Let $G\sim\sG(n,a/n,b/n)$ and assume, for some $\eps>0$,
$\lambda=(a-b)/\sqrt{2(a+b)} \ge 1+\eps$.
Then there exists $\Delta_{\est}=\Delta_{\est}(\eps)>0$ and $d_* =
d_*(\eps)>0$ such that, for all $d\ge d_*(\eps)$
\begin{align}
\prob\left(\frac{1}{n}|\<\xsdp(G),\bxz\>|\ge
  \Delta_{\est}(\eps)\right) \ge 1-C\, e^{-n^{1/2}/C}\, ,
\end{align}
with $\prob(\,\cdot\,)$ denoting expectation with respect to the
algorithm and the graph $G$, and $C=C(\eps)$ a constant.
\end{theorem}

\subsection{Robustness}

 Consider the problem of testing whether the
graph $G$ has a community structure, i.e. whether
$G\sim\sG(n,a/n,b/n)$ or $G\sim\sG(n,d/n)$, $d=(a+b)/2$.
The next result establishes that the  SDP-based test of Section
\ref{sec:MainPartition} is robust with respect to adversarial
perturbations of these models. Namely, an adversary can arbitrarily
modify $o(n)$ edges of these graphs, without changing the detection threshold.
\begin{corollary}\label{coro:Robustness}
Let $\prob_0$ the law of $G\sim\sG(n,d/n)$, and $\prob_1$ be the law
of $G\sim\sG(n,a/n,b/n)$. Denote by $\tprob_0$, $\tprob_1$ be any two
distributions over graphs with vertex set $V=[n]$. Assume that, for each
$a\in \{0,1\}$, the following happens: there exists a coupling
$\qprob_a$ of $\prob_a$ and $\tprob_a$ such that, if $(G,\tG)\sim
\qprob_a$, then $|E(G)\triangle E(\tG)|=o(n)$ with high probability. 

Then, under the same assumptions of Theorem \ref{thm:SDP_Test},
 the SDP-based test (\ref{eq:TestDef}) distinguishes $\tprob_0$ from
 $\tprob_1$ with error probability vanishing as $n\to\infty$.
\end{corollary}

By comparison,  spectral methods such as the one of
\cite{bordenave2015non} appear to be fragile to 
an adversarial perturbation of $o(n)$ edges \cite{OursReplicas}.

\subsection{Multiple communities} 

The hidden partition model of
Eq.~(\ref{eq:HiddenPart}) can be naturally generalized to the case of
$r>2$ hidden communities. Namely, we define the distribution
$\sG_r(n,a/n,b/n)$  over graphs as follows. The
vertex set $[n]$ is partitioned uniformly at random into $r$ subsets
$S_1$, $S_2$, \dots, $S_r$ with
$|S_i|=n/r$. Conditional on this partition, edges are independent with
\begin{align}
\prob_1\big((i,j)\in E|\{S_\ell\}_{\ell\le r}\big) = \begin{cases}
a/n & \mbox{ if $\{i,j\}\subseteq S_\ell$ for some $\ell\in[r]$,}\\
b/n & \mbox{ otherwise.}
\end{cases}\label{eq:rHiddenPart}
\end{align}
The resulting graph has average degree $d = [a+(r-1)b]/r$.
The case studied above  (hidden bisection) is recovered by setting
$r=2$ in this definition:  $\sG(n,a/n,b/n)= \sG_2(n,a/n,b/n)$. 
Of course, this model can be generalized further by allowing for $r$
unequal subsets, and a generic $r\times r$ matrix of edge
probabilities \cite{holland,abbe2015community,hajek2015achieving}.

Given a single realization of the graph $G$, we would like to test
whether  $G\sim\sG(n,d/n)$ (hypothesis $0$), or
$G\sim\sG_r(n,a/n,b/n)$ (hypothesis $1$). 
We use the same SDP relaxation already introduced in
Eq. (\ref{eq:SDP.DEF}), and the test  $T(\,\cdot\,;\delta)$ defined in Eq.~(\ref{eq:TestDef}).
This is particularly appealing because it does not require knowledge
of the number of communities $r$.
\begin{theorem}\label{thm:SDP_Test_r}
Consider the problem of distinguishing $G\sim\sG_r(n,a/n,b/n)$ from
$G\sim\sG(n,d/n)$, $d = (a+(r-1)b)/r$.
Assume, for some $\eps>0$,
\begin{align}
\frac{a-b}{\sqrt{r(a+(r-1)b)}} \ge 1+\eps\, .\label{eq:ConditionFactor1_r}
\end{align}
Then there exists $\delta_*=\delta_*(\eps,r)>0$ and $d_* = d_*(\eps,r)>0$
such that the following holds. If $d\ge d_*$, then the
SDP-based test $T(\,\cdot\,;\delta_*)$ succeeds 
with error probability probability  at most
$Ce^{-n/C}$ for $C=C(a,b,r)$ a constant.
\end{theorem}
\begin{remark}
In earlier work, a somewhat tighter relaxation is sometimes used,
including the additional constraint $X_{ij}\ge -(r-1)^{-1}$ for all
$i\neq j$. The simpler relaxation used here is however sufficient for
proving Theorem \ref{thm:SDP_Test_r}.
\end{remark}

\begin{remark}
The threshold established in Theorem \ref{thm:SDP_Test_r} coincides (for large degrees) with the
one of spectral methods using non-backtracking random walks \cite{bordenave2015non}. However, for $k\ge 4$
there appears to be a gap between 
general statistical tests and what is achieved by polynomial time algorithms \cite{decelle2011asymptotic,chen2014statistical}.
\end{remark}

\subsection*{Acknowledgments}

A.M. was partially supported by NSF grants CCF-1319979 and DMS-1106627 and the
AFOSR grant FA9550-13-1-0036. 
S.S was supported by 
the William R. and Sara Hart Kimball Stanford Graduate Fellowship.

\newpage

\bibliographystyle{amsalpha}

\newcommand{\etalchar}[1]{$^{#1}$}
\providecommand{\bysame}{\leavevmode\hbox to3em{\hrulefill}\thinspace}
\providecommand{\MR}{\relax\ifhmode\unskip\space\fi MR }
\providecommand{\MRhref}[2]{%
  \href{http://www.ams.org/mathscinet-getitem?mr=#1}{#2}
}
\providecommand{\href}[2]{#2}

%
%
%***************************************************************
%

\newpage

\appendix

\section{Proofs of Theorem \ref{thm:Main} and Theorem
  \ref{thm:SDP_Test} (main theorems)}
\label{sec:ProofMain}

In this Section we prove Theorem \ref{thm:Main} and Theorem
  \ref{thm:SDP_Test} using Theorems \ref{thm:Gro}, \ref{thm:Gaussian} and Lemmas
  \ref{lemma:ZeroTemperature}, \ref{lemma:Interpolation}. The proofs
  of the latter are presented in Appendices \ref{app:ProofGro},
  \ref{app:ProofZeroT}, \ref{sec:ProofInterpolation},
 \ref{app:ProofGaussianA}, \ref{app:ProofGaussianB}. 

We begin by proving a general approximation result, and then
obtain Theorem \ref{thm:Main} and Theorem
  \ref{thm:SDP_Test} as consequences.

\subsection{Three technical lemmas}
\label{sec:MainTech}
\begin{lemma}\label{lemma:ConcentrationPhi}
Let $G\sim\sG(n,a/n,b/n)$, $d=(a+b)/2$, and $\bAc_G = \bA_G-(d/n)\bone\bone^{\sT}$
be its centered adjacency matrix. For $\lambda\in \reals$ fixed,
define $\bB = \bB(\lambda)$ to be the deformed GOE matrix in Eq.~(\ref{eq:Bdefinition}).

Then, there exists a universal constant $C$ such
that, for either $\bM\in \{\bAc_G/\sqrt{d} ,\bB(\lambda)\}$,  for all $t\ge 0$
\begin{align}
\prob\Big\{\big| \Phi(\beta,k;\bM)-\E \Phi(\beta,k;\bM)\big|\ge
 nt\Big\}\le C\, e^{-nt^2/C}\, . \label{eq:ConcentrationPhi}
\end{align}
\end{lemma}
\begin{proof}
Define the following Gibbs probability measure over $(\bbS^{k-1})^n$,
which is naturally associated to the free energy $\Phi$:
\begin{align}
\mu_{\bM} (\bsigma) &\equiv\frac{ \exp(\beta H_\bM(\bsigma) )}{ \int
                      \exp(\beta H_\bM(\btau)) \de \nu(\btau) } \,
                      \de\nu(\bsigma)\, ,\\
H_{\bM}(\bsigma) &=
\<\bsigma,\bM\bsigma\> = \sum_{i,j=1}^n
M_{ij} \<\bsigma_i,\bsigma_j\>\, .
\end{align}
It is a straightforward exercise with moment generating functions to
show that
\begin{align}
\frac{\partial\Phi}{\partial M_{ij}} (\beta,k;\bM) &=
                                                   \mu_{\bM}(\<\bsigma_i,\bsigma_j\>)\, ,
\end{align}
where $\mu_{\bM}(f(\bsigma))$ denotes the expectation of $f(\bsigma)$ with respect to the probability
measure $\mu_{\bM}$. In particular, since
$|\<\bsigma_i,\bsigma_j\>|\le 1$ (here $\|\,\cdot\,\|_2$ denotes the
vector $\ell_2$ norm)
\begin{align}
\big\|\nabla_{\bM}\Phi\big\|_2^2 = \sum_{i,j=1}^n\left| \frac{\partial\Phi}{\partial M_{ij}}\right|^2\le n^2\, .
\end{align}
This implies Eq.~(\ref{eq:ConcentrationPhi}) for $\bM=\bB$ by Gaussian
isoperimetry (with constant $C=4$).

For $\bM = \bAc_G$ the proof is analogous. Let $G$ be a graph that
does not contain edge $(i,j)$, and $G^+$ denote the same graph, to
which edge $(i,j)$ has been added. Then writing the definition of
$\Phi(\,\cdots\,)$, we get
\begin{align}
\Phi\big(\beta,k;\bAc_{G^+}/\sqrt{d}\big)-\Phi\big(\beta,k;\bAc_{G}/\sqrt{d}\big)=\frac{1}{\beta}\log\Big\{\mu_{\bAc_{G}}
\big(e^{\frac{\beta}{\sqrt{d}}\<\bsigma_i,\bsigma_j\>}\big)\Big\}\, .
\end{align}
In particular 
\begin{align}
\Big|\Phi(\beta,k;\bAc_{G^+})-\Phi(\beta,k;\bAc_{G})\Big|\le
  \frac{1}{\sqrt{d}}\, .
\end{align}
The claim then follows from a standard application of the `method of
bounded differences' \cite{BLM} i.e. from Azuma-H\"oeffding
inequality, whereby we construct a bounded differences martingale
with a number of steps equal to a sufficiently large constant times
the number of edges, e.g. $10dn$. 
\end{proof}

\begin{lemma}\label{lemma:ConcentrationSDP}
Let $G\sim\sG(n,a/n,b/n)$, $d=(a+b)/2$, and $\bAc_G = \bA_G-(d/n)\bone\bone^{\sT}$
be its centered adjacency matrix. Then there exists a universal
constant $C$ such
that, for any $t\ge 0$
\begin{align}
\prob\Big\{\big| \SDP(\bAc_G)-\E \SDP(\bAc_G)\big|\ge
 nt\Big\}\le C\, e^{-nt^2/(Cd)}\, . \label{eq:ConcentrationSDP}
\end{align}
\end{lemma}
\begin{proof}
 Let $G$ be a graph that
does not contain edge $(i,j)$, and $G^+$ denote the same graph, to
which edge $(i,j)$ has been added. Let $\bX\in\PSD_1(n)$ be an
optimizer of the SDP with data $\bAc_G$, i.e. a feasible point such
that  $\<\bAc_G,\bX\> = \SDP(\bAc_G)$. Then
\begin{align}
\SDP(\bAc_{G^+})&\ge \<\bAc_{G^+},\bX\> \\
& = \<\bAc_{G},\bX\> +X_{ij}\\
&\ge \SDP(\bAc_G) -1\, ,
\end{align}
where we used the fact that $\bX$ is positive semidefinite to obtain 
$|X_{ij}|\le \sqrt{X_{ii}X_{jj}} = 1$. Exchanging the role of $G$ and
$G^+$, we obtain
\begin{align}
\big|\SDP(\bAc_{G^+})-\SDP(\bAc_G)\big|\le 1\, ,
\end{align}
As in the previous lemma, the claim follows from an application of the `method of
bounded differences' \cite{BLM} i.e. from Azuma-H\"oeffding inequality
(we can apply this to a martingale with a number of steps proportional
to the expected number of edges, say $10dn$, whence the claimed
probability bound follows).  
\end{proof}

\begin{lemma}\label{lemma:MixedNorm}
Let $\bAc_G$, $\bB$ be defined as in Lemma
\ref{lemma:ConcentrationPhi}. Then, there exists an absolute constant
$C>0$  such that the following holds with probability at least $1-C\, e^{-n/C}$:
\begin{align} 
\|\bAc_G\|_{\infty\to 2}\le C d\sqrt{n}\, ,\;\;\;\;\;\;\;
  \|\bB\|_{\infty\to 2}\le (C+\lambda)\sqrt{n}
\end{align}
\end{lemma}
\begin{proof}
For $\bB$ we use (letting $\|\bM\|_{2\to 2}= \|\bM\|_{op} =\max(\lambda_1(\bM),-|\lambda_n(\bM)|)$):
\begin{align}
\|\bB\|_{\infty\to 2}&\le \sqrt{n}\|\bB\|_{2\to 2}\le \sqrt{n} \big(\lambda
  + \|\bW\|_{2\to 2}\big)\\
&\le (C+\lambda)\sqrt{n}\, ,
\end{align}
where the last inequality holds with the desired probability by
standard concentration bounds on the extremal eigenvalues of GOE
matrices \cite{AGZ}[Section 2.3]. 

For $\bAc_G$, first note that 
\begin{align}
\|\bAc_G\|_{\infty\to 2}&\le
\|\bA_G\|_{\infty\to 2}+\frac{d}{n}\|\bone\bone^{\sT}\|_{\infty\to 2}\le \|\bA_G\|_{\infty\to 2}+\frac{d}{\sqrt{n}} \|\bone\bone^{\sT}\|_{2\to 2}\\
&\le \|\bA_G\|_{\infty\to 2}+d\sqrt{n}\, .
\end{align}
Next we observe that $\bsigma\mapsto\|\bA_G\bsigma\|_2^2$ is a
convex function on $\|\bsigma\|_{\infty} \leq 1$, and thus attains it maxima
at one of the corners of the hypercube $[-1,1]^n$. In other words,
$\|\bA_G\|_{\infty\to 2}^2 = \max_{\bsigma \in \{\pm
  1\}^n} \|\bA_G\bsigma\|_2^2$. For $\bsigma \in \{+1,-1 \}^n$, we get
\begin{align}
\|\bA_G\bsigma\|_2^2 \leq \sum_{i=1}^n \deg_G(i)^2 
\end{align}
where $\deg_G(i)$ is the degree of vertex $i$ in $G$. 
The desired bound follows since $\sum_{i=1}^n \deg(i)^2 \le C_0d^2
n$ with the desired probability for some constant $C_0$ large
enough (see, e.g. \cite{Janson}).
\end{proof}

\subsection{A general approximation result}

\begin{theorem}\label{thm:Approx}
Let $G\sim\sG(n,a/n,b/n)$, $d=(a+b)/2$, and $\bAc_G = \bA_G-(d/n)\bone\bone^{\sT}$
be its centered adjacency matrix. Let $\lambda = (a-b)/\sqrt{2(a+b)}$
and define $\bB = \bB(\lambda)$ to be the deformed GOE matrix in Eq.~(\ref{eq:Bdefinition}).
Then, there exists $C=C(\lambda)$ such that, with probability at least
$1-C\, e^{-n/C}$, for all $n\ge n_0(a,b)$
\begin{align}
\left|\frac{1}{n\sqrt{d}}\SDP(\bAc_G)-\frac{1}{n}\SDP(\bB(\lambda))\right|&\le
  \frac{C\log d}{d^{1/10}}\, ,\label{eq:SDPcontPlus}\\
\left|\frac{1}{n\sqrt{d}}\SDP(-\bAc_G)-\frac{1}{n}\SDP(-\bB(\lambda))\right|&\le
  \frac{C\log d}{d^{1/10}}\, . \label{eq:SDPcontMinus}
\end{align}
Further $C(\lambda)$ is bounded over compact intervals $\lambda\in
[0,\lambda_{\rm max}]$
\end{theorem}
\begin{proof}
Throughout the proof $C=C(\lambda)$ is a constant that depends
uniquely on $\lambda$, bounded as in the statement, and we will write `for $n$ large enough'
whenever a statement holds for $n\ge n_0(a,b)$.

First notice that by Lemma \ref{lemma:Interpolation} and Lemma \ref{lemma:ConcentrationPhi} we have, with
probability larger than $1-Ce^{-n/C}$, and all $n$ large enough,
\begin{align}
\left|\frac{1}{n}\Phi\big(\beta,k;\bAc_G/\sqrt{d}\big)-\frac{1}{n}\Phi(\big(\beta,k;\bB(\lambda)\big)\right|\le
\frac{4\beta^2}{\sqrt{d}} +\frac{10\lambda^{1/2}}{d^{1/4}}\, . \label{eq:BoundWHPPhi}
\end{align}
Next, by Lemma \ref{lemma:ZeroTemperature} and Lemma
\ref{lemma:MixedNorm},
with the same probability, for $\bM\in\{\bAc_G/\sqrt{d},\bB(\lambda)\}$, and $\beta,d>1$
\begin{align}
\Big|\frac{1}{n}\Phi(\beta,k;\bM)-\frac{1}{n}\OPT_k(\bM)\Big|\le
\frac{k}{\beta}\log\Big(\frac{C\beta(d+\lambda)}{k}\Big)
\end{align}
(where we optimized the bound of Lemma \ref{lemma:ZeroTemperature}
over $\eps$.)
Using triangle inequality with Eq.~(\ref{eq:BoundWHPPhi}), and
optimizing over $\beta$, we get, always with probability at least $1-Ce^{-n/C}$,
\begin{align}
\left|\frac{1}{n\sqrt{d}}\OPT_k(\bAc_G)-\frac{1}{n}\OPT_k(\bB(\lambda))\right|\le
  \frac{Ck^{2/3}}{d^{1/6}}\log\big(d+\lambda\big)\, .\label{eq:OPTkCont}
\end{align}
Proceeding the same way (with $\beta$ replaced by $-\beta$), we also
obtain
\begin{align}
\left|\frac{1}{n\sqrt{d}}\OPT_k(-\bAc_G)-\frac{1}{n}\OPT_k(-\bB(\lambda))\right|\le
  \frac{Ck^{2/3}}{d^{1/6}}\log\big(d+\lambda\big)\, .
\end{align}
Since $|\OPT_k(-\bB)|,|\OPT_k(-\bB)|\le n\|\bB\|_{op}\le Cn$
with probability at least $1-Ce^{-n/C}$, we get also 
\begin{align}
\max\left\{\frac{1}{n}\OPT_k(\pm\bB),\frac{1}{n\sqrt{d}}\OPT_k(\pm\bAc_G)\right\}\le
  C\, ,
\end{align}
whence, using Theorem \ref{thm:Gro}, we obtain
\begin{align}
\left|\frac{1}{n\sqrt{d}}\SDP(\bAc_G)-\frac{1}{n\sqrt{d}}\OPT_k(\bAc_G)\right|&  \le  \frac{C}{k}\,  ,\\
\left|\frac{1}{n}\SDP(\bB)-\frac{1}{n}\OPT_k(\bB)\right|& \le
                                                          \frac{C}{k}\,
                                                          .
\end{align}
The claim (\ref{eq:SDPcontPlus}) follows from using this, together
with Eq.~(\ref{eq:OPTkCont}) and triangular inequality. Equation
(\ref{eq:SDPcontMinus}) follows from exactly the same argument.
\end{proof}

\subsection{Proofs of Theorem  \ref{thm:Main}}

Applying Theorem \ref{thm:Gaussian} to $\lambda=0$ (whence $\bB(\lambda)
=\bW\sim\GOE(n)$), we get, with high probability,
\begin{align}
\frac{1}{n}\SDP(\bW), \frac{1}{n}\SDP(-\bW) \in
  \big[2-d^{-1},2+d^{-1}\big]\, .
\end{align}
(The claim for $-\bW$ follows because $-\bW\sim\GOE(n)$)
Using Theorem \ref{thm:Approx}, applied to $a=b=d$ (whence
$G\sim\sG(n,d/n)$), we have, with high probability
\begin{align}
\frac{1}{n\sqrt{d}}\SDP(\bAc_G), \frac{1}{n\sqrt{d}}\SDP(-\bAc_G) \in
  \left[2-\frac{C\log d}{d^{1/10}},2+\frac{C\log d}{d^{1/10}}\right]\, .
\end{align}
This implies that desired claim (\ref{eq:MaxLimit}) holds  with high
probability.
By the concentration lemma \ref{lemma:ConcentrationSDP} (with $a=b=d$)
it also 
holds with probability at least $1-C(d)e^{-n/C(d)}$.

\subsection{Proofs of Theorem  \ref{thm:SDP_Test}}

Recall --throughout the proof-- that $\lambda = (a-b)/\sqrt{2(a+b)}\ge
1+\eps$
and $d= (a+b)/2$. Further, without loss of generality, we can assume
$\lambda\in [0,\lambda_{\rm max}]$ with $\lambda_{\rm max}>1$ fixed
(e.g. $\lambda_{\rm max}=10^3$).

Recall that $\prob_0$ denotes the law of $G\sim \sG(n,d/n)$ and
$\prob_1$  the law of $G\sim \sG(n,a/n,b/n)$.  We can control the
probability of false positives (i.e. declaring $G$ to have a
two-communities structure, which it has not) using Theorem
\ref{thm:Main}. For any $\delta>0$, we have
\begin{align}
\lim_{n\to\infty}\prob_0\big(T(G;\delta) =1 \big) =
  \lim_{n\to\infty}\prob_0\Big(\frac{1}{n}\SDP(\bAc_G)\ge
  2(1+\delta)\sqrt{d} \Big) = 0\, ,
\end{align}
where the last equality holds for any $d\ge d_0(\delta)$.

We next bound the probability of false negatives. Let $\Delta(\,\cdot\,)$  as
per Theorem \ref{thm:Gaussian}. 
By Theorem \ref{thm:Approx}, there
exists $d_0'=d_0'(\eps)$ such that, for all $d\ge d_0'(\eps)$,
with high probability for $G\sim\sG(n,a/n,b/n)$,
\begin{align}
\frac{1}{n\sqrt{d}}\SDP(\bAc_G)&\ge \frac{1}{n}\SDP(\bB(\lambda)) -
                         \frac{1}{4}\Delta(1+\eps)\\
&\ge \frac{1}{n}\SDP(\bB(1+\eps)) -
                         \frac{1}{4}\Delta(1+\eps)\\
&\ge 2+\frac{3}{4}\Delta(1+\eps)\, ,\label{eq:LowerBoundFinalFalseNeg}
\end{align}
where the second inequality follows because  $\SDP(\bB(\lambda))$ is
monotone non-decreasing in $\lambda$ and the last inequality follows from  Theorem \ref{thm:Gaussian}. 

Selecting $\delta_*(\eps) = \Delta(1+\eps)/2>0$, we then have
\begin{align}
\lim_{n\to\infty}\prob_1\big(T(G;\delta_*(\eps)) =0 \big) &=
  \lim_{n\to\infty}\prob_1\Big(\frac{1}{n \sqrt{d}}\SDP(\bAc_G)<
  2(1+\delta_*(\eps)) \Big) \label{eq:ProbFalseNeg}\\
  & = \lim_{n\to\infty}\prob_1\Big(\frac{1}{n \sqrt{d}}\SDP(\bAc_G)<
  2+\Delta_*(1+\eps) \Big) = 0\, ,
\end{align}
where the last equality follows from
Eq.~(\ref{eq:LowerBoundFinalFalseNeg}).

We proved therefore that the error probability vanishes as
$n\to\infty$, provided $d>d_*(\eps) = \max(d_0(\delta_*(\eps)), \, d_0'(\eps))$.
In fact, our argument also implies (eventually adjusting $d_*$)
\begin{align}
\lim_{n\to\infty}\prob_0\Big(\frac{1}{n\sqrt{d}}\SDP(\bAc_G)\ge
  2+\frac{\delta_*}{2}\Big) &= 0\, ,\\
\lim_{n\to\infty}\prob_1\Big(\frac{1}{n\sqrt{d}}\SDP(\bAc_G)\le
  2+\delta_*\Big) &= 0\, .
\end{align}
It then follows from the concentration lemma
\ref{lemma:ConcentrationSDP} that these probabilities (and hence the
error probability of our test) are bounded by
$C\,e^{-n/C}$ for $C=C(a,b)$ a constant.
%
%***************************************************************
%
\section{Proof of Theorem \ref{thm:Regular} (SDP for random regular graphs)}
\label{app:Regular}

Recall that $\SDP(\bM)\le n\xi_1(\bM)$. Further, the leading
eigenvector of a $d$-regular graph is the the all-ones vector $\bv_1= \bone/\sqrt{n}$.
Using this remark together  almost-Ramanujan property of random $d$-regular graphs
\cite{Friedman}, 
we have, with high probability, 
\begin{align}
\frac{1}{n} \SDP(\bAc_G) \leq \xi_1(\bAc_G) =\xi_2(\bA_G) = 2
  \sqrt{d-1} + o_n(1)\, , 
\end{align}
This gives us the required upper bound. 

To derive a matching lower bound, we construct explicitly a feasible
point of the optimization problem which  asymptotically attains this
value as $n\to\infty$. 

To this end, let $T_d$ denote the infinite $d$-regular tree with
vertex set $V(T_d)$. 
Cs\'{o}ka et. al. \cite[Theorem 3,4]{csokaetal2014independent}
establish that for any $\lambda$ with 
$|\lambda| \leq d$, there exists a centered  Gaussian process indexed
by the vertices of $T_d$, $\{Z_v: v \in V(T_d)\}$, such that with
probability 1, for all $v \in V(T_d)$, 
\begin{align}
\sum_{u \in N(v)} Z_u = \lambda Z_v, 
\end{align}
where $N(u)$ denotes the neighbors of $u \in V(T_d)$. These processes
are referred to as ``Gaussian wave functions'',
Further, Cs\'{o}ka et. al. prove that for any $|\lambda| < 2
\sqrt{d-1}$, the process $\{Z_v: v \in V(T_d)\}$ can be approximated
by linear factor of i.i.d. processes. More explicitly, let  $\{X_v: v
\in V(T_d)\}$,
a collection of i.i.d. standard Gaussian $Y_v\sim\normal(0,1)$, then
there exists a sequence of coefficients $\{\alpha_{\ell}\}_{\ell\ge
  0}$, $\alpha_{\ell}\in\reals$ such that  the Gaussian wave function
$\{Z_v: v \in V(T_d)\}$ can be constructed so that
\begin{align}
&\lim_{L\to\infty}\E\left\{\Big(Z_v-Z^{(L)}_v\Big)^2\right\}=0\, ,\label{eq:L2Conv}\\
&Z_v^{(L)}\equiv \sum_{\ell=0}^L\sum_{u\in V(T_d): d(u,v) =
  \ell} \alpha_{\ell} Y_u\, .
\end{align}
(Here $d(\,\cdot\,,\,\cdot\,)$ is the usual graph distance.)

 We use  this construction with $\lambda = 2\sqrt{d-1}-\eps$ for
 $\eps$ a small positive number.
 Without loss of generality, we assume that $\Var(Z_v)=1$ for
 all $v\in V(T_d)$. It is easy to see \cite[Equation
 2]{csokaetal2014independent}
that for $u,v \in V(T_d)$ such that $(u,v)\in E(T_d)$, we have
\begin{align}
\E\{Z_u Z_v\} = \frac{2\sqrt{d-1}-\eps}{d}.  \nonumber 
\end{align}
Thus, denoting by $\dv$ the set of neighbors of vertex $v$, $\sum_{u
  \in \dv} \E\{Z_u Z_v\} = 2\sqrt{d-1}-\eps$. 
By Eq.~(\ref{eq:L2Conv}), there exists $L=L(\eps)$ large enough so that
\begin{align}
\sum_{u
  \in \dv} \E\{Z_u^{(L)} Z_v^{(L)}\} \ge 2\sqrt{d-1}-2\eps\, .
\end{align} 

Let $G\sim\Greg(n,d)$ be a random $d$-regular graph on $n$ vertices.
We use the above construction to obtain a feasible point of the SDP,
$\bX\in\PSD_1(n)$, with the desired value.
Namely, let $\{\tY_v:v\in V(G)\}$ be a collection of i.i.d. random
variables $\tY_v\sim\normal(0,1)$, independent of the graph $G$. We define $\{\tZ_v\, : v\in
V(G)\}$ using the same coefficients as above:
\begin{align}
\tZ_v^{(L)} = \sum_{k=0}^{L} \sum_{u\in V(G): d(u,v) = k} \alpha_{k}
  \tY_u, 
\end{align}
We then construct the matrix $\bX= (X_{ij})_{1\leq i,j \leq n}$ by letting
\begin{align}
X_{ij} =\frac{\E\{\tZ_i^{(L)}\tZ_j^{(L)}|G\}
  }{\sqrt{\E\{(\tZ_i^{(L)})^2|G\}\E\{(\tZ_j^{(L)})^2|G\}}}\, .
\end{align}
It is immediate to see from the construction that $\bX\in \PSD_1(n)$ is a feasible
point. 

At this feasible point, 
\begin{align}
\frac{1}{n}\< \bAc_G, X\> = \frac{1}{n} \sum_{i\in V(G)}\sum_{j\in \di }\frac{\E\{\tZ_i^{(L)}\tZ_j^{(L)}|G\}
  }{\sqrt{\E\{(\tZ_i^{(L)})^2|G\}\E\{(\tZ_j^{(L)})^2|G\}}}
-\frac{d}{n^2} \sum_{i,j\in V(G) }\frac{\E\{\tZ_i^{(L)}\tZ_j^{(L)}|G\}
  }{\sqrt{\E\{(\tZ_i^{(L)})^2|G\}\E\{(\tZ_j^{(L)})^2|G\}}}\, . 
\end{align}
Since $G$ converges almost surely as $n\to\infty$ to a $d$-regular tree (in the sense of local weak
convergence, see, e.g. \cite{dembo2010gibbs}),  and $\tZ^{(L)}_i$ is only a function of the
$L$-neighborhood of $i$, we have, $G$-almost surely
\begin{align}
\lim_{n\to\infty}\frac{1}{n} \sum_{i\in V(G)}\sum_{j\in \di }\frac{\E\{\tZ_i^{(L)}\tZ_j^{(L)}|G\}
  }{\sqrt{\E\{(\tZ_i^{(L)})^2|G\}\E\{(\tZ_j^{(L)})^2|G\}}} = \sum_{u\in \dv }\frac{\E\{Z_v^{(L)}Z_u^{(L)}\}
  }{\sqrt{\E\{(Z_v^{(L)})^2\}\E\{(Z_u^{(L)})^2\}}}\ge
  2\sqrt{d-1}-2\eps\, .
\end{align}
Also, since $\E\{\tZ_i^{(L)}\tZ_j^{(L)}\}=0$ , whenever $d(i,j)>2L$,
we have 
\begin{align}
\lim_{n\to\infty}\frac{d}{n^2} \sum_{i,j\in V(G) }\frac{\E\{\tZ_i^{(L)}\tZ_j^{(L)}|G\}
  }{\sqrt{\E\{(\tZ_i^{(L)})^2|G\}\E\{(\tZ_j^{(L)})^2|G\}}} = 0\, .
\end{align}
We conclude by noting that
\begin{align}
\lim_{n\to\infty}\frac{1}{n}\SDP(\bAc_G) \ge \lim_{n\to\infty}
  \frac{1}{n}\<\bAc_G,\bX\> \ge 2\sqrt{d-1}-2\eps\, ,
\end{align}
and the thesis follows since $\eps$ is arbitrary.

The proof for $-\bAc_G$ is exactly the same. 
%
%***************************************************************
%
\section{Proof of Theorem \ref{thm:Gro} (Grothendieck-type
  inequality)}
\label{app:ProofGro}

As mentioned already, the upper bound in Eq.~(\ref{eq:Gro}) is trivial.
The proof of the lower bound follows Rietz's method \cite{rietz1974proof}.

Let $\bX$ be a solution of the problem (\ref{eq:SDP.DEF}) and
through its Cholesky decomposition write $X_{ij} = \<\bsigma_i,\bsigma_j\>$, with 
$\bsigma_i\in \reals^n$, $\|\bsigma_i\|_2=1$. In other words we have, letting $\bM =
(M_{ij})_{i,j\in [n]}$, 
\begin{align}
\SDP(\bM) = \sum_{i,j=1}^n B_{ij}\<\bsigma_i,\bsigma_j\>\, .\label{eq:QBrepresentation}
\end{align}
Let $\Ga \in\reals^{k\times n}$ be a matrix with i.i.d. entries
$\Ga _{ij}\sim\normal(0,1/k)$. Define, $\bx_i\in\reals^k$, for $i\in [n]$, by letting
\begin{align}
\bx_i = \frac{\Ga \, \bsigma_i}{\|\Ga \, \bsigma_i\|_2}\, .
\end{align}
We next need a technical lemma.
\begin{lemma}\label{eq:RoundingLemma}
Let $\bu,\bv\in\reals^n$ with $\|\bu\|_2 = \|\bv\|_2 =1$ and
$\Ga \in\reals^{k\times n}$ be defined as above. Further, for
$\bw\in\reals^n$, let $z(\bw) \equiv
(1-\alpha_k^{-1/2}\|\Ga \bw\|_2^{-1})\Ga \bw$. Then
\begin{align}
\E\Big\<\frac{\Ga \bu}{\|\Ga \bu\|_2},\frac{\Ga \bv}{\|\Ga \bv\|_2}\Big\>= 
\alpha_k\<\bu,\bv\>+ \alpha_k\E\<z(\bu),z(\bv)\>\, .
\end{align}
\end{lemma}
\begin{proof}
Let $\bg_1,\bg_2\sim \normal(0,\id_{k}/k) $ be independent vectors
(distributed as the first two columns of $\Ga $. Let $a=\<\bu,\bv\>$ and
$b = \sqrt{1-a^2}$. Then by rotation invariance
\begin{align}
\E\<\Ga \bu,\Ga \bv\> = \E\<\bg_1, a\bg_1+\bg_2\> = a\E(\|\bg_1\|_2^2) = \<\bu,\bv\>\, ,
\end{align}
and 
\begin{align}
\E\big\<\frac{\Ga \bu}{\|\Ga \bu\|_2},\Ga \bv\big\> &= \E\big\<\frac{\bg_1}{\|\bg_1\|_2},
  a\bg_1+\bg_2\big\>\\
& = a\E(\|\bg_1\|_2) = 
\alpha_k^{1/2}\<\bu,\bv\>\, .
\end{align}

By expanding the product we have
\begin{align}
\E \<z(\bu),z(\bv)\> &= \<\bu,\bv\> -\alpha_k^{-1/2}
  \E\big\<\frac{\Ga \bu}{\|\Ga \bu\|_2},\Ga \bv\big\> 
-\alpha_k^{-1/2}
  \E\big\<\Ga \bu,\frac{\Ga \bv}{\|\Ga \bv\|_2}\big\>
                   +\frac{1}{\alpha_k}\E\big\<\frac{\Ga \bu}{\|\Ga \bu\|_2},\frac{\Ga \bv}{\|\Ga \bv\|_2}\big\>\\
&= -\<\bu,\bv\> + \frac{1}{\alpha_k}\E\big\<\frac{\Ga \bu}{\|\Ga \bu\|_2},\frac{\Ga \bv}{\|\Ga \bv\|_2}\big\>
\end{align}
which is equivalent to the statement of our lemma.
\end{proof}

Now, by definition of the $\bx_i$'s we have
\begin{align}
\E\Big\{\sum_{i,j=1}^n M_{ij}\<\bx_i,\bx_j\>\Big\}& = 
\sum_{i,j=1}^n M_{ij}
  \E\Big\<\frac{\Ga \bu_i}{\|\Ga \bu_i\|_2},\frac{\Ga \bu_j}{\|\Ga \bu_j\|_2}\Big\>\\
&  =  \alpha_k\sum_{i,j=1}^n M_{ij}\<\bu_i,\bu_j\>
+\alpha_k \sum_{i,j=1}^n M_{ij}\E\<z(\bu_i),z(\bu_j)\>\\
&  =  \alpha_k\SDP(\bM)
+\alpha_k\sum_{i,j=1}^n M_{ij}\E\<z(\bu_i),z(\bu_j)\>
\, .\label{eq:zs}
\end{align}
Now we interpret $z(\bu_i)$ as a vector in a Hilbert space with scalar
product $\E\<\,\cdot\,,\,\cdot\,\>$. Further by the rounding lemma
\ref{eq:RoundingLemma}, these vectors have norm 
\begin{align}
\E (\|z(\bu_i)\|_2^2) = \frac{1}{\alpha_k}-1\, .
\end{align}
Hence, by definition of $\SDP(\,\cdot\,)$,
we have
\begin{align}
-\sum_{i,j=1}^n M_{ij}\E\<z(\bu_i),z(\bu_j)\>\le
  \Big(\frac{1}{\alpha_k}-1\Big)\SDP(-\bM) \, .
\end{align}
Substituting this in Eq.~(\ref{eq:zs}), we obtain
\begin{align}
\OPT_k(\bM)\ge \E\Big\{\sum_{i,j=1}^n M_{ij}\<\bx_i,\bx_j\>\Big\}
\ge\alpha_k\SDP(\bM)
-(1-\alpha_k)\SDP(-\bM)\, ,
\end{align}
which coincides with the claim (\ref{eq:Gro}).

In order to prove Eq.~(\ref{eq:GroBis}), we apply  Eq.~(\ref{eq:Gro})
to $-\bM$, thus getting
\begin{align}
\SDP(-\bM)\le \frac{1}{\alpha_k}\OPT_k(-\bM) + \frac{1-\alpha_k}{\alpha_k}\,
  \SDP(\bM)\, .
\end{align}
Substituting this in Eq.~(\ref{eq:Gro}), we obtain
Eq.~(\ref{eq:GroBis}).

%
%*********************************************************************
%
\section{Proof of Lemma  \ref{lemma:ZeroTemperature}
  (zero-temperature approximation)}
\label{app:ProofZeroT}

Define the objective function
$H_{\bM}:(\bbS^{k-1})^n\to \reals$
\begin{align}
H_{\bM}(\bsigma) =
\<\bsigma,\bM\bsigma\> = \sum_{i,j=1}^n
M_{ij} \<\bsigma_i,\bsigma_j\>\, .
\end{align}
(In the first expression that $\<\,\cdot\,,\,\cdot\,\>$ denotes the scalar product between matrices
and we interpret $\bsigma$ as a matrix $\bsigma\in\reals^{n\times k}$.)
Let
$\bsigma^*\in\arg\max\{H_{\bM}(\bsigma):\;(\bbS^{k-1})^n\}$. 
We then have (denoting by
$\|\,\cdot\,\|_F$ the Frobenius norm):
\begin{align}
|H_{\bM}(\bsigma) - H_{\bM}(\bsigma^*)| &\le  | \< \bsigma- \bsigma^*,\bM\bsigma\>| + | \< \bsigma- \bsigma^*, \bM \bsigma^* \>| \\
&\le 2\, \max \{ \|\bM \bsigma\|_{F} ,
  \|\bM \bsigma^*\|_{F} \} \,\|\bsigma - \bsigma^*\|_{F} \\
& \le 2\, \sqrt{k}\,\|\bM\|_{\infty\to 2}\, \|\bsigma -
\bsigma^*\|_{F} \, .
\end{align}
Define the partition function
\begin{align}
Z(\beta,k;\bM)\equiv \int \,
                    \exp\big\{\beta H_{\bM}(\bsigma)\big\}\,\de\nu(\bsigma) \, ,
\end{align}
so that, in particular $\Phi(\beta,k;\bM) = (1/\beta)\log Z(\beta,k;\bM)$. By the above bound, 
and recalling $L\ge \|\bM\|_{\infty\to 2}/\sqrt{n}$
\begin{align}
e^{\beta H_{\bM}(\bsigma^{*})} \geq
Z(\beta,k;\bM)\geq e^{\beta
  H_{\bM}(\bsigma^{*})} 
\int \exp(-2\beta L\sqrt{kn} \,\|\bsigma -
\bsigma^*\|_{F}) \, \de \nu(\bsigma) \, .
\label{equation:boundsTemp}
\end{align}   
For any $\ve >0$, we have (here $\ind(\,\cdot\,)$ denotes the indicator function)
\begin{align}
\int \exp(-2\beta L \sqrt{kn} \|\bsigma - \bsigma^*\|_{F}) \,
\de \nu(\bsigma)  &\geq \int \exp(-2\beta L \sqrt{kn} \|\bsigma
- \bsigma^*\|_{F}) \ind(\max_{i\in [n]}\|\bsigma_i - \bsigma_i^*\|_2
\leq \ve )\, \de \nu(\bsigma)\nonumber \\
&\geq \exp(-2\beta Ln\ve \sqrt{k}  )\, ( V_k(\ve))^n \, ,  \label{equation:sphericalcap}
\end{align}
where $V_k(\ve)$ is  volume of the spherical cap $\{\bsigma_1 \in \bbS^{k-1}: \|\bsigma_1^* - \bsigma_1\|_2 \leq \ve\}$
(with respect to the normalized measure on the unit sphere $\bbS^{k-1}$).
By a simple integral in spherical coordinates  have $V_k(\ve) = (1/2) \prob\{ X < \ve^2 - (\ve^4/4)\}$ where $X\sim \dBet ( \frac{k-1}{2} , \frac{1}{2})$. Further
\begin{align}
\prob\Big( X < \ve^2 - \frac{\ve^4}{4}\Big) \geq \frac{1}{\dBet(\frac{k-1}{2}, \frac{1}{2}) } \int_{0}^{\ve^2 - \ve^4/4} t^{\frac{k-1}{2} -1} \de t 
\geq \frac{c}{\sqrt{k} } (\ve^2 - \ve^4/4)^{\frac{k-1}{2}} 
\end{align}
Plugging this Eq.~(\ref{equation:boundsTemp}), we obtain (since $\OPT_k(\bM) = H_{\bM}(\bsigma^*)$):
\begin{align}
e^{\beta \OPT_k(\bM)} \geq
Z(\beta,k;\bM)\geq e^{\beta
  \OPT_k(\bM)-2\beta Ln\ve \sqrt{k}} \left(\frac{\eps}{C}\right)^{kn} \, .
\end{align}   
Taking logarithms yields the desired bound (\ref{eq:TemperatureBound}).

%
%*********************************************************************
%
\section{Proof of Lemma \ref{lemma:Interpolation} (interpolation)}
\label{sec:ProofInterpolation}

Throughout this proof, we will fix, without loss of generality $S_1=\{1,\dots n/2\}$ and $S_2 =
\{(n/2)+1,\dots,n\}$. Define $\bv \in\reals^n$ by letting $v_i = 1/\sqrt{n}$ if $i\in S_1$ and $v_i = -1/\sqrt{n}$
if $i\in S_2$. Define
\begin{align}
\bBp(\lambda) = \lambda\bv\bv^{\sT} + \bW\, .
\end{align}
(We will drop the argument $\lambda$ when clear from the context.)
By a change of variables in the definition of $\Phi(\beta,k;\,\cdot\,)$ (namely, $\bsigma_i\to -\bsigma_i$ for
$i\in S_2$), and since $W_{ij}\ed -W_{ij}$, we have
\begin{align}
\E\Phi\big(\beta,k;\bB(\lambda)\big) = \E\Phi\big(\beta,k;\bBp(\lambda)\big)\, . 
\end{align}
We can and will therefore replace $\bB(\lambda)$ by $\bBp(\lambda)$. We will 
drop the superscript `new.'

We proceed in two steps, and define an intermediate Gaussian
random matrix
\begin{align}
\bD(\lambda) = \lambda\bv\bv^{\sT} + \bU\, ,
\end{align}
where $\bU = \bU^{\sT}\in\reals^{n\times n}$ is a Gaussian random
matrix with $\{U_{ij}\}_{1\le i\le j\le n}$ independent zero-mean Gaussian random
variables with
 \begin{align}
\Var(U_{ij}) = \begin{cases}
a[1-a/n]/(nd)& \mbox{ if $\{i,j\}\subseteq S_1$ or $\{i,j\}\subseteq S_2$,}\\
b[1-b/n]/(n d) & \mbox{ if $i\in S_1,\, j\in S_2$ or $i\in
  S_2 ,j\in S_1$,}
\end{cases}\label{eq:UnequalVariances}
\end{align}
and $U_{ii}=0$.
By triangular inequality
\begin{align}
\left|\frac{1}{n}\E\Phi\big(\beta,k;\bAc_G/\sqrt{d}\big)-\frac{1}{n}\E\Phi\big(\beta,k;\bB\big)\right|&\le 
\left|\frac{1}{n}\E\Phi\big(\beta,k;\bAc_G/\sqrt{d}\big)-\frac{1}{n}\E\Phi\big(\beta,k;\bD\big)\right|\nonumber\\
&+
\left|\frac{1}{n}\E\Phi\big(\beta,k;\bD\big)-\frac{1}{n}\E\Phi\big(\beta,k;\bB\big)\right|
\, .
\end{align}

The proof of  Lemma \ref{lemma:Interpolation} follows therefore from the next two results, which will be proved
in the next subsections.
\begin{lemma}\label{lemma:Interpolation1}
With the above definitions, if $n\ge (15d)^2$, then
\begin{align}
\left|\frac{1}{n}\E\Phi\big(\beta,k;\bAc_G/\sqrt{d}\big)-\frac{1}{n}\E\Phi\big(\beta,k;\bD\big)\right|\le \frac{2\beta^2}{\sqrt{d}}\, .
\end{align}
\end{lemma}

\begin{lemma}\label{lemma:Interpolation2}
With the above definitions, there exists an absolute constant $n_0$
such that, for all $n\ge n_0$, 
\begin{align}
\left|\frac{1}{n}\E\Phi\big(\beta,k;\bB\big)-\frac{1}{n}\E\Phi\big(\beta,k;\bD\big)\right|\le 5\sqrt{\frac{a-b}{d}}\, .
\end{align}
\end{lemma}

%
%***********
%
\subsection{Proof of Lemma \ref{lemma:Interpolation1}}

We use the following Lindeberg interpolation lemma, see e.g.  \cite{tao2012topics,chatterjee2005simple}.
\begin{lemma}\label{lemma:Lindeberg}
Let $F:\reals^N\to\reals$ be three times continuously differentiable.
Further, let $\bX = (X_1,\dots,X_N)$ and  $\bZ = (Z_1,\dots, Z_N)$ be two vectors of 
independent random variables, satisfying $\E\{X_i\} = \E\{Z_i\}$, $\E\{X_i^2\} = \E\{Z_i^2\}$
for each $i\in \{1,\dots, N\}$. Then, we have
\begin{align}
\big|\E\big\{F(\bX)-F(\bZ)\big\} \big|&\le  \frac{1}{6}\, S_3\, \max_{i\in [N]}\|\pai^3
  F\|_{\infty}\, , \\
S_3 &\equiv\sum_{i=1}^N\Big\{\E[|X_i|^3] +\E[|Z_i|^3]\Big\}\, .\label{eq:T3Def}
\end{align} 
where $\pai^{\ell} F(\bx)\equiv \frac{\partial^{\ell} F}{\partial x_i^{\ell}}$, and $\|\pai^{\ell} F\|_{\infty}\equiv \sup_{\bx\in\reals^N}|\pai^{\ell}F(\bx)|$.
\end{lemma}

We apply this to the function $\bM\mapsto\Phi(\beta,k;\bM)$ with $N= n(n-1)/2$, to compare the the two 
sets of independent random variables $\bD = \{D_{ij}\}_{i<j}$ and $\bM = \{M_{ij}\}_{i<j}$ where $\bM = \bAc/\sqrt{d}$. 
It is immediate to check the equality of the first two moments. Indeed
 \begin{align}
\E\{D_{ij}\} =\E\{M_{ij}\} =\begin{cases}
(a-b)/(2n\sqrt{d})& \mbox{ if $\{i,j\}\subseteq S_1$ or $\{i,j\}\subseteq S_2$,}\\
-(a-b)/(2n\sqrt{d}) & \mbox{ if $i\in S_1,\, j\in S_2$ or $i\in
  S_2 ,j\in S_1$,}
\end{cases}
\end{align}
and
 \begin{align}
\Var(D_{ij}) =\Var(M_{ij})=\begin{cases}
a[1-a/n]/(nd)& \mbox{ if $\{i,j\}\subseteq S_1$ or $\{i,j\}\subseteq S_2$,}\\
b[1-b/n]/(n d) & \mbox{ if $i\in S_1,\, j\in S_2$ or $i\in
  S_2 ,j\in S_1$,}
\end{cases}
\end{align}

Next we compute the partial derivatives of $\bM\mapsto\Phi(\beta,k;\bM)$. To this end,
it is convenient to define the following Gibbs probability measure over $(\bbS^{k-1})^n$,
which is naturally associated to the free energy $\Phi$:
\begin{align}
\mu_{\bM} (\bsigma) &= \frac{ \exp(\beta H_\bM(\bsigma) )}{ \int \exp(\beta H_\bM(\btau)) \de \nu(\btau) } \, \de\nu(\bsigma)\, .
\end{align}
where
\begin{align}
H_{\bM}(\bsigma) =
\<\bsigma,\bM\bsigma\> = \sum_{i,j=1}^n
M_{ij} \<\bsigma_i,\bsigma_j\>\, .
\end{align}
(The same construction was useful in Section \ref{sec:MainTech}. We
repeat it here for the reader's convenience.)
It is then immediate to get (letting $\partial_{ij}\equiv \frac{\partial\phantom{M}}{\partial M_{ij}}$):
\begin{align}
\partial_{ij}\Phi(\beta,k;\bM) &= \mu_{\bM}(\<\bsigma_i,\bsigma_j\>)\, ,\\
\partial^2_{ij}\Phi(\beta,k;\bM) &= \beta\Big(\mu_{\bM}(\<\bsigma_i,\bsigma_j\>^2)-\mu_{\bM}(\<\bsigma_i,\bsigma_j\>)^2\Big)\, ,\\
\partial^3_{ij}\Phi(\beta,k;\bM) &= \beta^2\Big(\mu_{\bM}(\<\bsigma_i,\bsigma_j\>^3)-3 \mu_{\bM}(\<\bsigma_i,\bsigma_j\>^2) \mu_{\bM}(\<\bsigma_i,\bsigma_j\>)
+2 \mu_{\bM}(\<\bsigma_i,\bsigma_j\>)^3\Big)
\, ,
\end{align}
where we used the convention of letting $\mu_{\bM}(f(\bsigma))$ denote the expectation of $f(\bsigma)$ with respect to the probability
measure $\mu_{\bM}$. In particular, the above imply
\begin{align}
\big\|\partial^3_{ij}\Phi\big\|_{\infty}\le 6\beta^2\, .\label{eq:ThirdDerivative}
\end{align}

We are finally left with the task of bounding the sum of third moments defined in Eq. (\ref{eq:T3Def}). Note that
$M_{ij} = (1-(d/n))/\sqrt{d}$  if $(i,j)\in E(G)$ and $M_{ij} = -\sqrt{d}/n$  otherwise.
Hence, we have
\begin{align}
\E\{|M_{ij}|^3\} \le \begin{cases}
(a/n)d^{-3/2}+(\sqrt{d}/n)^3& \mbox{ if $\{i,j\}\subseteq S_1$ or $\{i,j\}\subseteq S_2$,}\\
(b/n)d^{-3/2}+(\sqrt{d}/n)^3 & \mbox{ if $i\in S_1,\, j\in S_2$ or $i\in
  S_2 ,j\in S_1$,}
\end{cases}
\end{align}
Therefore
\begin{align}
S_3 &= \sum_{1\le i<j\le n}\E\big(|D_{ij}|^3\big) + \sum_{1\le i<j\le n}\E\big(|M_{ij}|^3\big) \\
&\le 
\frac{n^2}{2} \, 4\,\E\left\{\Big(\frac{\lambda}{n}\Big)^3+\Big(\frac{a}{n}\Big)^{3/2}|Z|^3\right\}
+ \frac{n^2}{4}\left\{\frac{a}{nd^{3/2}}+\left(\frac{\sqrt{d}}{n}\right)^3\right\}
+ \frac{n^2}{4}\left\{\frac{b}{nd^{3/2}}+\left(\frac{\sqrt{d}}{n}\right)^3\right\}\\
& \le
\frac{2\lambda^3}{n}+4n^{1/2}a^{3/2}+\frac{n}{2d^{1/2}}+\frac{d^{3/2}}{2n}\\
& \le 5 n^{1/2}a^{3/2}+ \frac{n}{d^{1/2}} \le \frac{2n}{\sqrt{d}}\, , \label{eq:S3Bound}
\end{align}
where the last two inequalities hold for $n\ge (15d)^2$. 

Finally, using Lemma \ref{lemma:Lindeberg} with
Eq.~(\ref{eq:ThirdDerivative}) and the bound (\ref{eq:S3Bound}) we
obtain 
\begin{align}
\big|\E\Phi(\beta,k;\bM)-\E\Phi(\beta,k;\bD)\big|\le \frac{2\beta^2
  n}{\sqrt{d}}\, ,
\end{align}
which is the required claim.
%
%***********
%
\subsection{Proof of Lemma \ref{lemma:Interpolation2}}

This proof is by coupling. We first observe that (here the scalar
product  $\<\bsigma,\bM\bsigma\>$ is to be interpreted as a product
between matrices with $\bsigma\in\reals^{n\times k}$)
\begin{align}
\Phi(\beta,k;\bB) & =  \frac{1}{\beta}\, \log\left\{\int \,
                    \exp\big\{\beta\<\bsigma,\bB\bsigma\>\big\}\,\de\nu(\bsigma)\right\}\\
&=  \frac{1}{\beta}\, \log\left\{\int \,
                    \exp\big\{\beta\<\bsigma,\bD\bsigma\>\
                    +\beta\<(\bB-\bD),\bsigma\bsigma^{\sT}\>\big\}\,\de\nu(\bsigma)\right\}\\
&\le  \frac{1}{\beta}\, \log\left\{\int \,
                    \exp\big\{\beta\<\bsigma,\bD\bsigma\>\
                    +n\beta\|\bB-\bD\|_{op}\big\}\,\de\nu(\bsigma)\right\}\\
& \le \Phi(\beta,k;\bB)  + n \|\bB-\bD\|_{op}\, ,
\end{align}
where we used $\|\bsigma\bsigma^{\sT}\|_* = \|\bsigma\|_F^2 = n$ (with
$\|\,\cdot\,\|_*$ denoting the nuclear norm). Hence
\begin{align}
\Big|\frac{1}{n}\Phi(\beta,k;\bB)
-\frac{1}{n}\Phi(\beta,k;\bD)\Big|\le \|\bB-\bD\|_{op}\, . \label{eq:MatrixCont}
\end{align}
In order to couple $\bB$ and $\bD$ se construct three independent
symmetric Gaussian random matrices $\bZ_0$, $\bZ_1$,
$\bZ_2\in\reals^{n\times n}$ as follows. All of the three matrices
have centered independent entries, differ in the
variances. Setting $v(a) =(a/(nd))\, (1-a/n)$, and $v(b)
=(b/(nd))\, (1-b/n)$, we let
\begin{align}
\Var(Z_{0,ij}) & = \begin{cases}
v(b) & \mbox{ if $i\neq j$,}\\
0& \mbox{ if $i=j$,}
\end{cases}\\
\Var(Z_{1,ij}) &= \begin{cases}
v(a)-v(b) & \mbox{ if $\{i,j\}\subseteq S_1$ or $\{i,j\}\subseteq
  S_2$, and $i\neq j$,}\\
0& \mbox{otherwise,}
\end{cases}
\end{align}
and, finally,
\begin{align}
\Var(Z_{2,ij}) & = \begin{cases}
(1/n)-v(b) & \mbox{ if $i\neq j$,}\\
(1/n)& \mbox{ if $i=j$.}
\end{cases}
\end{align}
It is therefore easy to see that
\begin{align}
\bB & = \lambda\, \bv\, \bv^{\sT} +  \bZ_0 + \bZ_2\, ,\\
\bD & = \lambda\, \bv\, \bv^{\sT} +  \bZ_0 + \bZ_1\, .
\end{align}
Hence using Eq.~(\ref{eq:MatrixCont}) and triangular inequality
\begin{align}
\Big|\frac{1}{n}\E\Phi(\beta,k;\bB)
-\frac{1}{n}\E\Phi(\beta,k;\bD)\Big|& \le \E\|\bZ_1\|_{op}+
\E\|\bZ_2\|_{op}\\
&\le 2.1\sqrt{1-nv(b)}+ 2.1\sqrt{n(v(a)-v(b))/2} \\
&\le 5\sqrt{\frac{a-b}{d}}
 \, ,
\end{align}
where the last bounds hold for all $n\ge n_0$
by standard estimates  on the eigenvalues of GOE matrices
\cite{AGZ}.
%
%*********************************************************************
%
\section{Proof of Theorem \ref{thm:Gaussian}.$(a)$ (deformed GOE
  matrices, $\lambda\le 1$)}
\label{app:ProofGaussianA}

In this section we prove part $(a)$ of  Theorem \ref{thm:Gaussian}. We
start with two useful technical facts, and then present the actual
proof. Throughout $\bB(\lambda) = (\lambda/n)\, \bone\bone^{\sT}+
\bW$, with $\bW\sim\GOE(n)$
is defined as per Eq.~(\ref{eq:Bdefinition}).

\subsection{Two technical lemmas}

\begin{lemma}\label{lemma:Monotoniciy}
For any fixed $\bW$, the function $\lambda\mapsto \SDP(\bB(\lambda))$ is monotone nondecreasing.
\end{lemma}
\begin{proof}
Let $\lambda_1\le \lambda_2$ and choose $\bX_*\in\PSD_1(n)$ such that
$\<\bB(\lambda_1),\bX_*\> = \SDP(\bB(\lambda_1))$ (this exists since
$\PSD_1(n)$ is compact). Then 
\begin{align}
\SDP(\bB(\lambda_2)) &\ge \<\bB(\lambda_2),\bX_*\> \\
& \ge
  \<\bB(\lambda_1)+(\lambda_2-\lambda_1)\bone\bone^{\sT}/n,\bX_*\>\\
&\ge \SDP(\bB(\lambda_1)) \, ,
\end{align}
where the last inequality follows since $\bX_*\succeq 0$.
\end{proof}

\begin{lemma}\label{lemma:BoundProj}
Fix $\delta\in (0,1]$ and  $k(n)=\lfloor n\delta\rfloor$. Let
$\bU\in\reals^{n\times k(n)}$ be a uniformly random (Haar measure) orthogonal matrix
(in particular $\bU^{\sT}\bU= \id_{k(n)}$). Then there exists
$C=C(\delta)$ such that, for any fixed basis vector $\bfe_i$, 
\begin{align}
\prob\Big(\max_{1\le i\le n}\big| \|\bU^{\sT}\bfe_i\|_2^2-\delta\big|\ge
C\sqrt{\frac{\log n}{n}}\Big)\ge 1-\frac{1}{n^{20}}\, .
\end{align}
\end{lemma}
\begin{proof}
In order to lighten the notation, we can assume $n\delta$ to be an integer.

Let $\bP=\bU\bU^{\sT}$ be the orthogonal projector on the column space
of $\bU$. 
By the invariance of the Haar measure under rotations, this 
is a projector onto a uniformly random subspace of $n\delta$ dimension
in $\reals^n$, and $Y_i \equiv \|\bU^{\sT}\bfe_i\|_2^2= \<\bfe_i,\bP\bfe_i\> = \|\bP\bfe_i\|_2^2$. Inverting the role of $\bP$
and $\bfe_i$, we see that $Y_{ii}$ is distributed as the square norm of
the first $n\delta$ components of a uniformly random unit vector of $n$
dimensions.
Hence
\begin{align}
Y_{i} \ed \frac{Z_{n\delta}}{Z_{n\delta}+Z_{n(1-\delta)}}\, ,
\end{align}
where $Z_{\ell}\sim \chi^2(\ell)$, $\ell\in\{n\delta,n(1-\delta)\}$ denote
two independent chi-squared random variable with $\ell$
degrees of freedom.
Standard tail bounds on chi-squared random variables imply the claim.
\end{proof}

\subsection{Proof of Theorem \ref{thm:Gaussian}.$(a)$}

We first note that
\begin{align}
\frac{1}{n}\SDP(\bB(\lambda)) \le \xi_1(\bB(\lambda))\le 2+o_n(1)\, ,
\end{align}
where the last inequality holds with high probability, by, e.g., \cite{knowles2013isotropic}[Theorem 2.7]. 

It is therefore sufficient to prove that, for any $\eps>0$,
$\SDP(\bB(\lambda))/n\ge 2-\eps$ with probability converging to one as
$n\to\infty$. By Lemma \ref{lemma:Monotoniciy}, we only need to prove
this for $\lambda =0$, i.e. to lower bound $\SDP(\bW)$ for $\bW\sim\GOE(n)$. We will
achieve this by constructing a witness, i.e. a feasible point
$\bX\in\PSD_1(n)$,  depending on $\bW$ such that $\<\bW,\bX\>/n\ge
2-\eps$ with high probability. 

A more general construction will be developed in Appendix
\ref{app:ProofGaussianB} to prove part $(b)$ of the Theorem. The case
$\lambda=0$ is however much simpler and we prefer to present it
separately here to build intuition.

Fix $\delta>0$, and let $\bu_1$, $\bu_2$,\dots $\bu_{n\delta}$ be the
eigenvectors of $\bW$ corresponding to the top $n\delta$
eigenvalues. Denote by $\bU\in\reals^{n\times (n\delta)}$, $\bU^{\sT}\bU=\id_{n\delta}$ the matrix whose
columns are $\bu_1$, $\bu_2$,\dots $\bu_{n\delta}$, and let
$\bD\in\reals^{n\times n}$ be the diagonal matrix with entries
\begin{align}
\bD_{ii} = (\bU\bU^{\sT})_{ii}\, .
\end{align}
Note that, by invariance of the GOE distribution under orthogonal
transformations, $\bU$ is a uniformly random orthogonal matrix.
Hence by Lemma \ref{lemma:BoundProj} and union bound
\begin{align}
\prob\Big(\max_{i\in [n]}|D_{ii}-\delta|\le C\sqrt{\frac{\log
  n}{n}}\Big) \ge 1-\frac{1}{n^{9}}\, ,\label{eq:BoundDii}
\end{align}
for $C=C(\delta)$ a suitable constant.

We then define our witness as
\begin{align}
\bX = \bD^{-1/2}\bU\bU^{\sT} \bD^{-1/2}\, .
\end{align}
Clearly $\bX\in\PSD_1(\bW)$ is a feasible point. 
Further, letting $\bE = \delta^{1/2}\bD^{-1/2}$
\begin{align}
\<\bW,\bX\>& = \frac{1}{\delta}\<\bW,\bU\bU^{\sT}\>- \frac{1}{\delta}\<\bW-\bE\bW\bE,\bU\bU^{\sT}\>\\
& \ge \frac{1}{\delta}\sum_{\ell=1}^{n\delta}\xi_{\ell}(\bW) -
\frac{1}{\delta}\|\bW-\bE\bW\bE\|_2\|\bU\bU^{\sT}\|_*\\
& \ge n\xi_{n\delta}(\bW) -
\frac{1}{\delta}\|\bW\|_2(1+\|\bE\|_2)\|\bE-\id\|_2\|\bU\bU^{\sT}\|_*\, .
\end{align}
Here $\|\bZ\|_*$ denotes the nuclear norm of $bZ$ (sum of the absolute
values of eigenvalues) and in the last inequality we used
$\|\bW-\bE\bW\bE\|_2\le \|\bW-\bE\bW\|_2 +\|\bE\bW-\bE\bW\bE\|\le \|\bW\|_2\|\bE-\id\|_2 +
\|\bE\|_2\|\bW\|_2 \|\bE-\id\|_2$. 

Next , since $\bU\bU^{\sT}$ is a projector on $n\delta$ dimensions, we have
$\|\bU\bU^{\sT}\|_*=n\delta$, whence
\begin{align}
\frac{1}{n}\, \<\bW,\bX\>\ge \lambda_{n\eps}(\bW)-
\|\bW\|_2(2+\|\bE-\id\|_2)\|\bE-\id\|_2\, .\label{eq:QW_Lower}
\end{align}
By Eq.~(\ref{eq:BoundDii}), we have $\|\bE-\id\|_2\to 0$ almost surely,
and  by a classical result \cite{AGZ}, 
also the following limits hold almost surely
\begin{align}
\lim_{n\to\infty}\|\bW\|_2 & = 2\, ,\\
\lim_{n\to\infty} \lambda_{n\delta}(\bW) & =\xi_*(\delta)\, ,
\end{align}
where $\xi_*(\delta)\uparrow 2$ as $\delta\to 0$. Indeed $\xi_*(\delta)$ can be
expressed explicitly in terms of Wigner semicircle law, namely, for
$\delta\in (0,1)$  it is
the unique positive solution of the following equation.
\begin{align}
\int_{\xi_*(\delta)}^2 \frac{\sqrt{4-x^2}}{2\pi}\;\de x = \delta\, .\label{eq:XiStarDef}
\end{align}
Substituting in Eq.~(\ref{eq:QW_Lower}), we get, almost surely (and as
consequence in probability)
\begin{align}
\lim\inf_{n\to\infty}\frac{1}{n}\, \<\bW,\bX\> \ge \xi_*(\delta)\ge 2-\eps\, .
\end{align}
where the last inequality holds by taking $\delta$ small enough.
%
%*********************************************************************
%
\section{Proof of Theorem \ref{thm:Gaussian}.$(b)$ (deformed GOE
  matrices, $\lambda>1$)}
\label{app:ProofGaussianB}

We begin by recalling the definition of the deformed GOE matrix
$\bB=\bB(\lambda)$, given in  Eq.~(\ref{eq:Bdefinition}),
\begin{align}
\bB \equiv \frac{\lambda}{n}\, \bone\bone^{\sT}+ \bW\, ,
\end{align}
where $\bW\sim\GOE(n)$, and we denote by $(\bu_1,\xi_1)$, \dots,
$(\bu_n,\xi_n)$
denote the eigenpairs of $\bB$, namely
\begin{align}
\bB\bu_k = \xi_k\bu_k\, ,
\end{align}
where $\xi_1\ge \xi_2\ge \dots\ge \xi_n$.

The proof of Theorem  \ref{thm:Gaussian}.$(b)$ is based on the following
construction of a witness $\bX$, which depends on (small) parameters $\eps, \delta>0$ to be fixed
at the end. In order not to complicate the notation un-necessarily, we
will assume $n\delta$ to be an integer.
 Let $R:\reals\to \reals$ be a `capping' function, i.e.
\begin{align}
R(x) \equiv
\begin{cases}
1 & \mbox{ if $x \geq 1$,}\\
x & \mbox{ if $-1<x<1$,}\\
-1 & \mbox{ if $x\le -1$.}\\
\end{cases}
\end{align}
We then define $\bphi\in\reals^n$ by letting $\varphi_i \equiv
  R(\eps\sqrt{n} \, u_{1,i})$. 
We also define $\bU\in\reals^{n\times (n\delta)}$ as the matrix whose
  $i$-th column is $\bu_{i+1}$ (hence it contains the eigenvector
  $\bu_2,\dots \bu_{n\delta+1}$). Note that $\bU$ is an orthogonal
  matrix: $\bU^{\sT}\bU = \id_{n\delta}$.
Finally, we define $\bD\in\reals^{n\times n}$ to be a diagonal matrix
  with entries
\begin{align}
D_{ii} = \frac{\sqrt{1-\varphi^2_i}}{\|\bU^{\sT}\bfe_i\|_2}\, .
\end{align}
Our witness construction is defined as
\begin{align}
\bX = \bphi\bphi^{\sT} + \bD\bU\bU^{\sT}\bD\, .\label{eq:WitnessConstruction}
\end{align}

We analyze this construction through a sequence of lemmas. One  of the
proofs will use Lemma \ref{lemma:LLN_u}, to which we devote a separate
section.
Throughout
we assume the above definitions and the setting of 
  Theorem \ref{thm:Gaussian}.  We use $C, C_0,\dots$ to denote finite non-random
 universal  constants.  Without loss of generality, we will also assume
  $\lambda\in (1,C_0)$ for some $C_0>1$.

We start from an elementary fact.
\begin{lemma}\label{lemma:BnormBound}
There exists a constant $C$ such that
\begin{align}
\lim_{n\to\infty}\prob\big(\|\bB\|_2\ge C\big) = 0\, .\label{eq:Bnorm}
\end{align}
\end{lemma}
\begin{proof}
It follows from triangular inequality that $\|\bB\|_2\le
\lambda+\|\bW\|_2$. Hence the claim follows by standard bounds on the eigenvalues of GOE
matrices \cite{AGZ}[Theorem 2.1.22].
\end{proof}

\begin{lemma}\label{lemma:Bphi}
There exists a constant $C>0$ such that, with high probability,
\begin{align}
\Big|\frac{1}{n}\<\bB,\bphi\bphi^{\sT}\>-\eps^2\xi_1\Big|\le C\,
  \eps^4\, .
\end{align}
\end{lemma}
\begin{proof}
Define $x-R(x) \equiv \baR(x)$. Further, for a vector
$\bx=(x_1,x_2,\dots,x_n)$, we
write $\baR(\bx)$ for the vector obtained applying $\baR$
componentwise, i.e.  $\baR(\bx) =
(\baR(x_1),\baR(x_2),\dots,\baR(x_n))$.
We then have
\begin{align}
\Big|\frac{1}{n}\<\bB,\bphi\bphi^{\sT}\>-\eps^2\xi_1\Big| &=
\Big|\frac{1}{n}\<\bB,\bphi\bphi^{\sT}\>-\frac{1}{n}\<\bB,(\eps\sqrt{n}\bu_1)
(\eps\sqrt{n}\bu_1)^{\sT}\>\Big|\\
&\le
\frac{2}{n}\big|\<(\eps\sqrt{n}\bu_1),\bB\,\baR(\eps\sqrt{n}\bu_1)\>\big|+ 
\frac{1}{n}\big|\<\baR(\eps\sqrt{n}\bu_1),\bB\, \baR(\eps\sqrt{n}\bu_1)\>\big|\\
&\le 4\, \|\bB\|_2\frac{1}{\sqrt{n}}\big\|\baR(\eps\sqrt{n}\bu_1)\big\|_2\,
\max\Big(\eps\; ;\;
\frac{1}{\sqrt{n}}\big\|\baR(\eps\sqrt{n}\bu_1)\big\|_2\Big)\, .\label{eq:BoundsBat}
\end{align}

Note that 
\begin{align}
\baR(x)^2 = \begin{cases}
(|x|-1)^2 & \mbox{ if $|x|\ge 1$,}\\
0 & \mbox{ if $|x|<1$.}
\end{cases}
\end{align}
In particular $\baR(x)^2\le x^6$ for all $x$.  
We therefore have
\begin{align}
\frac{1}{n}\big\|\baR(\eps\sqrt{n}\bu_1)\big\|_2^2&=
\frac{1}{n}\sum_{i=1}^n \baR(\eps\sqrt{n} u_{1,i})^2\\
& \le \frac{\eps^6}{n}\sum_{i=1}^n (\sqrt{n} u_{1,i})^6\, .
\end{align}
Next we decompose $\bu_1 = z_1(\bone/\sqrt{n})+ \sqrt{1-z_1^2}\,
\bup_1$, where
$z_1 = |\<\bu_1,\bone\>|/\sqrt{n}\in [0,1]$, and $\<\bup_1,\bone\>=0$. Since
$(a+b)^6\le 2^5(a^6+b^6)$, we have
\begin{align}
\frac{1}{n}\big\|\baR(\eps\sqrt{n}\bu_1)\big\|_2^2&\le
\frac{\eps^6}{n}\sum_{i=1}^n 32\big(1+(\sqrt{n} \uper_{1,i})^6\big)\\
& \le 32\eps^6\left[1+\frac{1}{n}\sum_{i=1}^n (\sqrt{n}
  \uper_{1,i})^6\right]
\le C \eps^6\, ,
\end{align}
where the last inequality holds with high probability for some
absolute constant $C$ and all $n\ge n_0$, by Lemma \ref{lemma:LLN_u} below,
applied with $a=6$, $b=0$. Using this together with
Eq.~(\ref{eq:Bnorm}) in Eq.~(\ref{eq:BoundsBat}) we get
\begin{align}
\Big|\frac{1}{n}\<\bB,\bphi\bphi^{\sT}\>-\eps^2\xi_1\Big| &\le
C\eps^3\,\max(\eps;\; C\eps^3)\le C'\eps^4\, ,
\end{align}
which completes our proof.
\end{proof}

\begin{lemma}\label{lemma:DtoF}
Let $\bF\in\reals^{n\times n}$ be a diagonal matrix with entries
$F_{ii} = \sqrt{1-\varphi^2_i}$. Then, there exists a constant
$K=K(\delta)$ such that,  with high probability,
\begin{align}
\Big|\frac{1}{n}\<\bB,\bD\bU\bU^{\sT}\bD\>-\frac{1}{n\delta}\<\bB,\bF\bU\bU^{\sT}\bF\>\Big|
\le K(\delta)\, \sqrt{\frac{\log n}{n}}
\, .
\end{align}
\end{lemma}
\begin{proof}
Define $\bH$ to be a diagonal matrix with entries $H_{ii}\equiv
\sqrt{\delta}/\|\bU^{\sT}\bfe_i\|_2$.
Then by definition $\bD =  \bF\bH/\sqrt{\delta}$ and
\begin{align}
\Big|\frac{1}{n}\<\bB,\bD\bU\bU^{\sT}\bD\>-\frac{1}{n\delta}\<\bB,\bF\bU\bU^{\sT}\bF\>\Big|
&=
\frac{1}{n\delta}\big|\<\bF\bB\bF,\bH\bU\bU^{\sT}\bH\>-\<\bF\bB\bF,\bU\bU^{\sT}\>\Big|\\
& \le \frac{1}{n\delta}
\big\|\bH\btB\bH-\btB\big\|_2\|\bU\bU^{\sT}\|_*\, ,
\end{align}
where $\btB = \bF\bB\bF$, and we recall that $\|\bM\|_*$ denotes the
nuclear norm of matrix $\bM$. Note that $\|\bF\|_2= \max_{i\in
  [n]}|F_{ii}|\le 1$, hence by Eq.~(\ref{eq:Bnorm}) we have
$\|\btB\|_2\le C$ with high probability. Further, since $\bU\bU^{\sT}$
is a projector on a space of $n\delta$ dimensions, we have
$\|\bU\bU^{\sT}\|_* = n\delta$. 
Therefore 
\begin{align}
\Big|\frac{1}{n}\<\bB,\bD\bU\bU^{\sT}\bD\>-\frac{1}{n\delta}\<\bB,\bF\bU\bU^{\sT}\bF\>\Big|
&\le \big\|\bH\btB\bH-\btB\big\|_2\\
&\le \|\btB\|_2\|\bH-\id\| \, \big(2+\|\bH-\id\|_2\big)\\
&\le C\|\bH-\id\| \max(1\;;\; \|\bH-\id\|_2\big)\, ,
\end{align}
where we used $\|\btB\|_2 \le \|\bB\|_2\|\bF\|_2^2 \le \|\bB\|_2\le C$
by Lemma \ref{lemma:BnormBound}.
Note that 
\begin{align}
\|\bH-\id\|_2 = \max_{1\le i\le n}
\Big|\frac{\sqrt{\delta}}{\|\bU^{\sT}\bfe_i\|_2}-1\Big|\, .
\end{align}
The proof is completed by Lemma \ref{lemma:BoundProj} and union bound.
\end{proof}

\begin{lemma}\label{lemma:LB_ui}
There exists a finite constant $C>0$ such that, for all
$\delta,\eps>0$, we have
\begin{align}
\lim_{n\to\infty}\prob\Big(
\<\bu_i,\bF\bB\bF\bu_i\> \ge L(\eps,\delta)\, \forall i\in \{2,\dots, n\delta+1\}\Big) = 1\, ,\\
L(\eps,\delta) \equiv 2-2\eps^2- C\delta^{2/3}
  -C\eps^4\, .
\end{align}
\end{lemma}
The proof of this lemma is longer that the others, and
  deferred to Section \ref{sec:LB_ui}.

We are now in position to prove  Theorem \ref{thm:Gaussian}.$(b)$.
\begin{proof}[Proof of Theorem \ref{thm:Gaussian}.$(b)$]
We use the explicit construction in
Eq.~(\ref{eq:WitnessConstruction}). Note that
$\bX\in\PSD_1(n)$. Indeed $\bX\succeq 0$ as it is the sum of two
positive-semidefinite matrices. Further, $X_{ii}=1$, since
\begin{align}
\<\bfe_i,\bX\bfe_i\> & =
|\<\bfe_i,\bphi\>|^2+\big\|\bU^{\sT}\bD\bfe_i\big\|_2^2\\
& = \varphi_i^2+ D_{ii}^2 \|\bU^{\sT}\bfe_i\|_2^2 =1\, .
\end{align}

We are left with the task of lower bounding the objective value.
With high probability
\begin{align}
\frac{1}{n}\<\bB,\bX\> &=\frac{1}{n}\<\bB,\bphi\bphi^{\sT}\>
+\frac{1}{n}\<\bB,\bD\bU\bU^{\sT}\bD\>\\
& \ge \eps^2\xi_1 - C\, \eps^4
+\frac{1}{n\delta}\<\bB,\bF\bU\bU^{\sT}\bF\>-K(\delta)\sqrt{\frac{\log
  n}{n}}\, ,
\end{align}
where we used Lemma \ref{lemma:Bphi}, and Lemma
\ref{lemma:DtoF}. For all $n$ large enough,  we can bound the term $\sqrt{(\log
n)/n}^{1/2}$ by $C\eps^4$. Further, by
\cite{knowles2013isotropic}[Theorem 2.7],
$\xi_1\ge (\lambda+\lambda^{-1})- C'n^{-0.4}$ with high probability. Since $\lambda+\lambda^{-1}>2$,
there exists $\Delta_0(\lambda)>0$  such that, with high probability
\begin{align}
\frac{1}{n}\<\bB,\bX\> &\ge (2+\Delta_0(\lambda))\eps^2 -C\eps^4
+\frac{1}{n\delta}\sum_{i=2}^{n\delta+1} \<\bu_i,\bF\bB\bF\bu_i\>\, .
\end{align}
Now we apply Lemma \ref{lemma:LB_ui} to get, with high probability
\begin{align}
\frac{1}{n}\<\bB,\bX\> &\ge (2+\Delta_0(\lambda))\eps^2 -C\eps^4
+2-2\eps^2-C\delta^{2/3} -C\eps^4\\
& \ge 2 +\Delta_0(\lambda)\eps^2- 2C\eps^4-C\delta^{2/3}\, .
\end{align}
Setting $\eps = \sqrt{\Delta_0(\lambda)/(4C)}$ and $\delta
=[\Delta_0(\lambda)/(16C^2)]^{3/2}$, we conclude that
\begin{align}
\lim_{n\to\infty}\prob\left(\frac{1}{n}\<\bB,\bX\>\ge
  2+\frac{\Delta_0(\lambda)^2}{16C}\right) = 1\, ,
\end{align}
which completes the proof of the theorem.
\end{proof}

\subsection{A law of large numbers for the eigenvectors of deformed
  Wigner matrices}

In this section we establish a lemma that will be used repeatedly in 
the proof of Lemma  \ref{lemma:LB_ui}.
\begin{lemma}\label{lemma:LLN_u}
Fix $i\in\{2,\dots, n\}$ and let $\bup_1$, $\bup_i$ be the projections
of eigenvectors $\bu_1$, $\bu_i$ of $\bB$ orthogonal to $\bone$
(explicitly, $\bup = \bu-\<\bone,\bu\>\bone/n$ for $\bu\in
\{\bu_1,\bu_i\}$). For any $a, b\in \naturals$, and $t, C\in\reals_{>0}$ there exists $n_0
= n_0(a,b,t,C)<\infty$ such that, for all $n>n_0$
\begin{align}
\prob\left\{\left|\frac{1}{n}\sum_{k=1}^n(\sqrt{n}\uper_{1,k})^a(\sqrt{n}\uper_{i,k})^b-m_am_b\right|\ge
  t\right\}\le \frac{1}{n^C}\, ,
\end{align}
where $m_a \equiv \E\{Z^a\}$, for $Z\sim\normal(0,1)$.
\end{lemma}
\begin{proof}
Throughout the proof, we let $\bv \equiv \bone/\sqrt{n}$. Note that
the law of the random matrix $\bB$ is invariant under transformations
that leave $\bv$ unchanged. namely, if $\bR\in\reals^{n\times n}$ is
an orthogonal matrix such that $\bR\bv = \bv$ or $\bR\bv = -\bv$, then
\begin{align}
\bR\bB\bR^{\sT} \ed \bB\, .
\end{align}
It follows that the joint law of $\bup_1$, $\bup_i$ is left invariant
by such a transformation. Formally $(\bR\bup_1,\bR\bup_i) \ed
(\bup_1,\bup_i)$. Hence, the pair $(\bup_1,\bup_i)$ is a uniformly
random orthonormal pair, in the subspace orthogonal to $\bv$
(invariance under rotations characterizes this distribution uniquely).
Hereafter, we'll set $i=2$ without loss of generality.

We can construct the pair by generating i.i.d. vectors
$\bg_1,\bg_2\sim\normal(0,\id_n)$,
and then applying Gram-Schmidt procedure to the triple
$(\bv,\bg_1,\bg_2)$.
Explicitly
\begin{align}
\bup_1 & =
\frac{\bg_1-\<\bg_1,\bv\>\bv}{\|\bg_1-\<\bg_1,\bv\>\bv\|_2}\, ,\label{eq:BupDef}\\
\bup_2 & = \frac{\bg_2-\<\bg_2,\bv\>\bv-\<\bg_2,\bup_1\>\bup_1}{\|\bg_2-\<\bg_2,\bv\>\bv-\<\bg_2,\bup_1\>\bup_1\|_2}\, .
\end{align}
We then have 
\begin{align}
&\frac{1}{n}\sum_{k=1}^n(\sqrt{n}\uper_{1,k})^a(\sqrt{n}\uper_{i,k})^b
\equiv \frac{U_{a,b} }{U_{2,0}^{a/2} U_{0,2}^{b/2}}\, ,\label{eq:Urepresentation}\\
&U_{a,b}  \equiv \frac{1}{n}\sum_{k=1}^n
\big(g_{1,k}-\<\bg_1,\bv\>v_k\big)^a
\big(g_{2,k}-\<\bg_2,\bv\>v_k-\<\bg_2,\bup_1\>\uper_{1,k}\big)^b\, .\label{eq:UabDef}
\end{align}
We claim that, with the same notations as in the statement of the
lemma, 
\begin{align}
\prob\left\{\left|U_{a,b}-m_am_b\right|\ge
  t\right\}\le \frac{1}{n^C}\, ,\label{eq:ClaimUab}
\end{align}
for $n\ge n_0(a,b,t,C)$.
Once this claim is proved, the lemma follows by the representation
(\ref{eq:Urepresentation}) using union bound over the three random
variables $U_{a,b}$, $U_{2,0}$, $U_{0,2}$, since $m_2=1$ (and
eventually increasing $n_0$).

In order to prove the claim (\ref{eq:ClaimUab}), we expand the powers
in Eq.(\ref{eq:UabDef}), to get:
\begin{align}
U_{a,b} &= U_{a,b}(0) +\sum_{0\le l_1\le a}\sum_{0\le l_2,l_3\le b}
K_{a,b}(l_1,l_2,l_3)\, U_{a,b}(l_1,l_2,l_3)\, \bone_{l_1+l_2+l_3>0}\, 
\bone_{l_2+l_3\le b}\, ,\label{eq:UabSum}\\
 U_{a,b}(0) & \equiv \frac{1}{n}\sum_{k=1}^ng_{1,k}^ag_{2,k}^b\, ,\\
U_{a,b}(l_1,l_2,l_3) & \equiv
\frac{1}{n^{(l_1+l_2)/2}} \<\bg_1,\bv\>^{l_1}\<\bg_2,\bv\>^{l_2}
\<\bg_2,\bup_1\>^{l_3}
\Big(\frac{1}{n}\sum_{k=1}^ng_{1,k}^{a-l_1}g_{2,k}^{b-l_2-l_3} (\uper_{1,k})^{l_3}\Big)\, ,\label{eq:U123}
\end{align}
where $K_{a,b}(l_1,l_2,l_3)$ are combinatorial factors (bounded as
$|K_{a,b}(l_1,l_2,l_3)|\le 2^a3^b$).
Consider first the term $U_{a,b}(0)$. By definition $\E\{U_{a,b}(0)\}
= m_am_b$. Further, by Markov inequality, 
\begin{align}
\prob\left\{\left|U_{a,b}(0)-m_am_b\right|\ge t\right\}&\le
\frac{1}{t^{\ell}n^{2\ell}}\E\left\{\left[\sum_{i=1}^nX_i\right]^{2\ell}\right\}\\
&\le \frac{1}{t^{\ell}n^{2\ell}} \, n^{\ell}C_0(a,b,\ell)\le
\frac{1}{n^C}\, ,\label{eq:U0bound}
\end{align}
where $C_0$ is a combinatorial factor, and last inequality holds for
any $C$, provided $n\ge n_0(a,b,t,C)$.

Consider next any of the terms $U_{a,b}(l_1,l_2,l_3)$. Note that
$\<\bg_1,\bv\>$, $\<\bg_2,\bv\>$, $\<\bg_3,\bup_1\>\sim\normal(0,1)$ 
(but not independent). By Gaussian tail bounds,
$\prob(|\<\bg_1,\bv\>|\ge a\sqrt{\log n})\le n^{-a^2/4}$ for all $n$
large enough. By a union bound
\begin{align}
\prob \Big\{ |\<\bg_1,\bv\>|^{l_1}|\<\bg_2,\bv\>|^{l_2}
|\<\bg_2,\bup_1\>|^{l_3}\ge (\log n)^{a+b}\Big\}\le \frac{1}{n^C}\, ,
\end{align}
for all $C>0$, provided $n\ge n_0(C)$. 
Proceeding analogously, and using the construction (\ref{eq:BupDef}),
we get for all $n\ge n_0(C)$,
\begin{align}
\prob\Big\{(\uper_{1,k})^{l_3}\ge \left(\frac{\log
    n}{n}\right)^{l_3/2}\Big\}\le \frac{1}{n^C}\, .
\end{align}

Finally, using these probability bounds in Eq.~(\ref{eq:U123}), we get, with probability
at least $1-2n^{-C}$,
\begin{align}
|U_{a,b}(l_1,l_2,l_3)| & \le
\frac{1}{n^{(l_1+l_2)/2}} (\log n)^{a+b}
\Big(\frac{1}{n}\sum_{k=1}^ng_{1,k}^{2(a-l_1)}g_{2,k}^{2(b-l_2-l_3)}
(\uper_{1,k})^{2l_3}\Big)^{1/2}\\
& \le \frac{1}{n^{(l_1+l_2+l_3)/2}} (\log n)^{a+2b}
U_{2(a-l_1),2(b-l_2-l_3)}(0)^{1/2}\, .
\end{align}
Hence, using Eq.~(\ref{eq:UabSum}) and the bound
(\ref{eq:U0bound}) applied to $U_{2(a-l_1),2(b-l_2-l_3)}(0)$, we
obtain (since $l_1+l_2+l_3\ge 1$)
\begin{align}
\prob\Big(\big|U_{a,b}-U_{a,b}(0)\big|\ge
\frac{(\log n)^{a+b}}{n^{1/2}}\Big)\le \frac{1}{n^C}\, ,
\end{align}
for all $C>0$ and all $n\ge m_0(a,b,t,C)$. 
Applying again Eq.~(\ref{eq:U0bound}) to $U_{a,b}(0)$, we obtain the
desired bound, Eq.~(\ref{eq:ClaimUab}), which finishes the proof.
\end{proof}

\subsection{Proof of Lemma \ref{lemma:LB_ui}} 
\label{sec:LB_ui}

We begin with a technical lemma.
\begin{lemma}\label{lemma:ProjBulk}
Fix $i\in \{2,\dots,n\}$ and let $\bu_i$ be the $i$-th eigenvector of 
the deformed GOE matrix $\bB$. Let $\bv= \bone/\sqrt{n}$.

Then, for any $\eta>0$ there exists
$n_0=n_0(\eta)$ (independent of $i$) such that, for all $n\ge n_0(\eta)$
\begin{align}
\prob\Big(|\<\bv,\bu_i\>|\ge \eta\Big) \le \frac{1}{n^{10}}\, .
\end{align}
\end{lemma}
\begin{proof}
Consider the eigenvalue equation $\bB\bu_i = \xi_i\bu_i$ or, equivalently,
\begin{align}
\lambda\<\bv,\bu_i\>\bv + \bW\bu_i = \xi_i \bu_i\, .
\end{align}
Solving for $\bu_i$ and then using $\|\bu_i\|_2^2 = 1$, we get the
equation
\begin{align}
1 = \lambda^2\<\bu_i,\bv\>^2\, \<\bv,\big(\xi_i\id-\bW\big)^{-2}\bv\>\, . 
\end{align}
Since, by assumption $\lambda>1$, it is sufficient to prove that,
for any $M>0$,  $\<\bv,\big(\xi_i\id-\bW\big)^{-2}\bv\>\ge M$ with
probability at least $1-n^{-10}$ provided $n\ge n_0(M)$.

In order to prove this fact, let $(\xi_{0,1}, \bu_{0,1}),\dots, (\xi_{0,n}, \bu_{0,n}),$ be the
eigenpairs of $\bW$, and notice that, by the interlacing inequality 
$\xi_{0,i-1}>\xi_i>\xi_{0,i}$. Further assume $i\in\{2,\dots, n/2\}$
(the proof proceeds analogously in the other case). Then, fixing
$\sigma>0$  a small number, we have
\begin{align}
\<\bv,\big(\xi_i\id-\bW\big)^{-2}\bv\>&=
                                         \sum_{k=1}^n\frac{|\<\bv,\bu_{0,k}\>|^2}{(\xi_i-\xi_{0,k})^2}\\
&\ge
  \sum_{k=i+1}^{i+n\sigma}\frac{|\<\bv,\bu_{0,k}\>|^2}{(\xi_{0,i}-\xi_{0,i+n\sigma})^2}\\
& \ge \frac{1}{(\xi_{0,i}-\xi_{0,i+n\sigma})^2}\,
  \|\bU_0^{\sT}\bv\|_2^2\, ,\label{eq:Resolvent}
\end{align}
where, for notational simplicity, we assumed $n\sigma$ to be an
integer, and $\bU_0\in \reals^{n\times (n\sigma)}$ is a matrix
whose columns are the eigenvectors
$\bu_{0,i+1},\dots,\bu_{0,i+n\sigma}$.

Note that, by invariance of $\bW\sim\GOE(n)$ under rotations $\bU_0$
is a uniformly random orthogonal matrix with the assigned dimension.
By Lemma \ref{lemma:BoundProj} implies for all $n\ge n_1(\sigma)$,
\begin{align}
\prob\Big(\|\bU_0^{\sT}\bv\|_2^2\ge
\frac{\sigma}{2}\Big)\ge 1-\frac{1}{n^{20}}\, .\label{eq:U00bound}
\end{align}
For $k\in\{1,\dots, n\}$, let  $\oxi_k$ be the unique solution in
$(-2,2)$ of 
\begin{align}
\int_{\oxi_k}^2\frac{\sqrt{4-x^2}}{2\pi}\, \de x= \frac{k}{n}\, .
\end{align}
Then, concentration of the eigenvalues of Wigner matrices
\cite{AGZ}[Theorem 2.3.5], together with the convergence to the
semicircle law, implies, for all $n\ge n_2(\sigma)$, and  letting $j=i+n\sigma$,
\begin{align}
\prob\Big(|\xi_i-\oxi_i|\le \sigma,\,
  |\xi_{j}-\oxi_{j}|\le \sigma\Big) \ge
  1-\frac{1}{n^{20}}\, .\label{eq:ConcentrationTwoEvalues}
\end{align}
Further, by definition,
\begin{align}
\sigma &= \int_{\oxi_j}^{\oxi_i}\frac{\sqrt{4-x^2}}{2\pi}\, \de x\\
& \ge \int_{2-(\oxi_i-\oxi_j)}^{2}\frac{\sqrt{4-x^2}}{2\pi}\, \de x\\
& \ge C_0 \, (\oxi_i-\oxi_j)^{3/2}\,,
\end{align}
with $C_0$ a numerical constant.
Using this bound together with the concentration bound
(\ref{eq:ConcentrationTwoEvalues}) we get, for all $\sigma$ small
enough, and all $n\ge n_2(\sigma)$
\begin{align}
\prob\big(|\xi_i-\xi_{i+n\sigma}|\le C_1\,\sigma^{2/3}\big)\ge
  1-\frac{1}{n^{20}}\, .
\end{align}
Using this inequality together with Eq.~(\ref{eq:U00bound}) in
Eq.~(\ref{eq:Resolvent}), we get
\begin{align}
\prob\Big(\<\bv,\big(\xi_i\id-\bW\big)^{-2}\bv\>\ge
  C_2\sigma^{-1/3}\Big)\ge 1-\frac{1}{n^{10}}\, ,
\end{align}
which implies the claim of the Lemma, by taking $\sigma$ a small
enough constant.
\end{proof}

Define $\projp_{1,i}$ to be the projector orthogonal to the space spanned 
by $\{\bu_1,\bu_i\}$. The following Lemma bounds the contribution of
  this space.
\begin{lemma}\label{lemma:Perp}
Recall  that $\bF \in\reals^{n\times n}$ denotes the diagonal matrix
with entries $F_{ii} = \sqrt{1-\varphi_i^2}$.
Then, there exists  constants $C>0$, and $n_0=n_0(\eps)$ such that, for all $i\in
\{2,\dots,n\delta+1\}$, and all $n\ge n_0(\eps)$, we have
\begin{align}
\prob\big(\|\projp_{1,i}\bF\bu_i\|_2\ge C\eps^2\big)\le
 \frac{C}{n^4}\, .
\end{align}
\end{lemma}
\begin{proof}[Proof of Lemma \ref{lemma:Perp}]
We decompose $\bu_i$ as
\begin{align}
\bu_i = z_i\, \frac{\bone}{\sqrt{n}} +\sqrt{1-z_i^2}\, \bup_i
\end{align}
where $z_i = |\<\bu_i,\bone/\sqrt{n}\>| \in [0,1]$ and $\<\bup_i,\bone\>=0$
(note that we can assume $z_i\ge 0$ by eventually flipping $\bu_i$).
Since $\|\bF-\id\|_2 =\max_{1\le i\le n}|F_{ii}-1|\le 1$, and
$\projp_{1,i}\bu_i=0$, we have
\begin{align}
\|\projp_{1,i}\bF\bu_i\|_2 &=\|\projp_{1,i}(\bF-\id)\bu_i\|_2 \\
&\le z_i\,
  \|\projp_{1,i}(\bF-\id)\bone/\sqrt{n}\|_2 + \sqrt{1-z_i^2}
  \|\projp_{1,i}(\bF-\id)\bup_i\|_2\\
&\le z_i + \|(\bF-\id)\bup_i\|_2\, .\label{eq:ZiPlus}
\end{align}
From Lemma \ref{lemma:ProjBulk}, there exists a constant $n_1=n_1(\eps)$
such that, for all $n\ge n_1(\eps)$
\begin{align}
\prob\big(z_i\ge \eps^2\big) \le \frac{1}{n^5}\, .\label{eq:ZiBound}
\end{align}

For the second contribution in Eq.~(\ref{eq:ZiPlus}) we use
\begin{align}
\|(\bF-\id)\bup_i\|_2^2 & =
                          \sum_{k=1}^n\big(\sqrt{1-\varphi_k^2}-1\big)^2\big(\uper_{i,k}\big)^2\\
&\stackrel{(a)}{\le} \sum_{k=1}^n\varphi_k^4\big(\uper_{i,k}\big)^2\\
&\stackrel{(b)}{\le}  \eps^4n^2\sum_{k=1}^n\big(u_{1,k})^4\big(\uper_{i,k}\big)^2\label{eq:S_4}\\
&\stackrel{(c)}{\le}
  \eps^4\left(\frac{1}{n}\sum_{k=1}^n\big(\sqrt{n}\,
  u_{1,k})^8\right)^{1/2}\left(\frac{1}{n}\sum_{k=1}^n\big(\sqrt{n}\, \uper_{i,k}\big)^4\right)^{1/2}
\, ,\label{eq:BoundFMinusId}
\end{align}
where inequality $(a)$ follows from $1-\sqrt{1-t}\le t$ for $t\in
[0,1]$, inequality $(b)$ from $R(x)^2\le x^2$, and $(c)$ from
Cauchy-Schwartz.

We next bound with high probability each term on the right hand side
in Eq.~(\ref{eq:BoundFMinusId}). 
In the following, we let $\bv \equiv \bone/\sqrt{n}$.
Let us start with the second
term.  By applying Lemma \ref{lemma:LLN_u}, with $a=0$, $b=4$, we find
that, for all $n\ge n_0$ (with $n_0$ an absolute constant) 
\begin{align}
\prob\Big(\frac{1}{n}\sum_{k=1}^n\big(\sqrt{n}\, \uper_{i,k}\big)^4\ge 4\Big)\le \frac{1}{n^9}\, . \label{eq:BupBound}
\end{align}

Consider next the first term on the right-hand side of
Eq.~(\ref{eq:BoundFMinusId}). We have $\bu_1 = z_1\, \bv +
\sqrt{1-z_1^2}\bup_1$, where $z_1 = |\<\bu_1,\bv\>|\in [0,1]$, and -- again--
$\bup_1$ is orthogonal to $\bv$. By triangular
inequality, we have
$\|\bu_1\|_8 \le z_1\|\bv\|_8 + \sqrt{1-z_1^2}\|\bup_1\|_8 \le 
n^{-3/8} +\|\bup_1\|_8$, and therefore
\begin{align}
\frac{1}{n}\sum_{k=1}^n\big(\sqrt{n}\,
  u_{1,k})^8\le 128 + \frac{128}{n}\sum_{k=1}^n\big(\sqrt{n}\,
  \uper_{1,k})^8\, .\label{eq:TriangularBup}
\end{align}
Using this bound together with  Lemma \ref{lemma:LLN_u} (with $a=8$, $b=0$) we find
that, for all $n\ge n_0$ (with $n_0$ an absolute constant) 
\begin{align}
\prob\Big(\frac{1}{n}\sum_{k=1}^n\big(\sqrt{n}\, \uper_{1,k}\big)^8\ge 1000\Big)\le \frac{1}{n^9}\, . \label{eq:BupBound2}
\end{align}

Using Eqs.~(\ref{eq:BupBound}) and (\ref{eq:BupBound2}) in
Eq.~(\ref{eq:BoundFMinusId}), we
get, of
all $n$ large enough and some constant $C$, 
\begin{align}
\prob\Big(\|(\bF-\id)\bup_i\|_2\ge C\eps^2\Big)& \le \frac{1}{n^8}\,
,
\end{align}

Using this in Eq.~(\ref{eq:ZiPlus}), together with
Eq~(\ref{eq:ZiBound}), we obtain the desired claim.
\end{proof}
The next lemma controls the effect of $\bF$ along $\bu_i$. 
\begin{lemma}\label{lemma:Ui}
There exists  constants $C>0$, and $n_0=n_0(\eps)$ such that, for all $i\in
\{2,\dots,n\}$, and all $n\ge n_0(\eps)$, we have
\begin{align}
\prob\big(\<\bu_i,\bF\bu_i\>\ge \sqrt{1-\eps^2}-C\eps^4\big)\ge 1-
 \frac{C}{n^4}\, . \label{eq:ClaimLemmaUi}
\end{align}
\end{lemma}
\begin{proof}[Proof of Lemma \ref{lemma:Ui}]
Throughout the proof, we let $\bv\equiv\bone/\sqrt{n}$.
We decompose $\bu_i = z_i\bv + \sqrt{1-z_i^2}\,\bup_i$, where 
$z_i = |\<\bv,\bu_i\>|\in [0,1]$ and $\<\bv,\bup_i\>=0$ (note that we can
always assume $\<\bu_i,\bv\>\ge 0$ by eventually flipping $\bu_i$). 
Since $\bF$ is diagonal with $F_{ii} = \sqrt{1-\varphi_i^2}$, we have
$\|\bF\|_2= \max_{1\le i\le n}|F_{ii}|\le 1$, and $\bF\succeq
0$. Therefore 
\begin{align}
\<\bu_i,\bF\bu_i\>&= z_i^2\<\bv,\bF\bv\>+
  2z_i\sqrt{1-z_i^2}\<\bv,\bF\bup_i\>+(1-z_i^2) \<\bup_i,\bF\bup_i\>\\
&\ge \<\bup_i,\bF\bup_i\> - 2z_i-z_i^2 \\
&\ge \<\bup_i,\bF\bup_i\> - 3z_i\, ,
\end{align} 
It follows from  Lemma \ref{lemma:ProjBulk} that $z_i\le \eps^4/3$
with probability at least $1-n^{-10}$ for all $n$ large enough, and
any fixed $i\ge 2$.
Therefore, for all $n\ge n_0'(\eps)$, we
have that
\begin{align}
\prob\Big(\<\bu_i,\bF\bu_i\>\ge \<\bup_i,\bF\bup_i\> -\eps^4\Big)\ge 1-\frac{1}{n^{10}}\, .\label{eq:uFu-upFup}
\end{align}

We are now left with the task of lower bounding
$\<\bup_i,\bF\bup_i\>$.
By definition, we have
\begin{align}
\<\bup_i,\bF\bup_i\> &= \frac{1}{n}\sum_{k=1}^n\sqrt{1-\varphi_k^2} \;
  \big(\sqrt{n}\uper_{i,k}\big)^2\\
&\stackrel{(a)}{\ge} 1-\frac{1}{2n}\sum_{k=1}^n\varphi_k^2
  \big(\sqrt{n}\uper_{i,k}\big)^2-\frac{2}{n}\sum_{k=1}^n\varphi_k^4
  \big(\sqrt{n}\uper_{i,k}\big)^2\\
&\stackrel{(b)}{\ge} 1-\frac{\eps^2}{2n}\sum_{k=1}^n\big(\sqrt{n}\, u_{1,k}\big)^2
  \big(\sqrt{n}\uper_{i,k}\big)^2-\frac{2\eps^4}{n}\sum_{k=1}^n\big(\sqrt{n}\, u_{1,k}\big)^4
  \big(\sqrt{n}\uper_{i,k}\big)^2
\, . \label{eq:SumOneMinusEps}
\end{align}
where  inequality we $(a)$ follows since $\sqrt{1-x}\ge 1-(x/2)-2x^2$
for $x\in [0,1]$, and $(b)$ because $|R(x)|\le x$. 

We next consider each of the sums on the right-hand side of
Eq.~(\ref{eq:SumOneMinusEps}).
These take the form
\begin{align}
S_q\equiv \frac{1}{n}\sum_{k=1}^n\big(\sqrt{n}\, u_{1,k}\big)^{q}
  \big(\sqrt{n}\uper_{i,k}\big)^2\, ,
\end{align}
where $q=2$ (for the first sum) or $q=4$ (for the second). Using this
notation, we have
\begin{align}
\<\bup_i,\bF\bup_i\> &\ge 1-\frac{1}{2}\eps^2 S_2-2\eps^4 \, S_4\, .\label{eq:uFu}
\end{align}
The term $S_4$ has been already dealt with in the proof of Lemma
\ref{lemma:Perp},
see Eq.~(\ref{eq:S_4}).  By the same derivation, we conclude that
there exists an absolute constant $C$ such that 
\begin{align}
\prob\big(S_4\ge C\big) \le \frac{1}{n^8}\, ,\label{eq:probS4}
\end{align}
for all $n\ge n_0$. 

Next consider $S_2$. We decompose $\bu_1= z_1\, \bv+\sqrt{1-z_1^2}\,
\bup_1$ where $z_1 = |\<\bu_1,\bv\>|$ and
$\<\bup_1,\bv\>=0$. Expanding the square, and using $v_k =
1/\sqrt{n}$, we get
\begin{align}
S_2 
&=z_1^2\, \frac{1}{n}\sum_{k=1}^n
  \big(\sqrt{n}\uper_{i,k}\big)^2+
2z_1\sqrt{1-z_1^2}\, \frac{1}{n}\sum_{k=1}^n\big(\sqrt{n}\, \uper_{1,k}\big)
  \big(\sqrt{n}\uper_{i,k}\big)^2+
(1-z_1^2)\, \frac{1}{n}\sum_{k=1}^n\big(\sqrt{n}\, \uper_{1,k}\big)^2
  \big(\sqrt{n}\uper_{i,k}\big)^2\,.\label{eq:ExpansionS2}
\end{align}
Because of the invariance of the $\GOE$ distribution under 
orthogonal transformations, the pair $\{\bup_1,\bup_i\}$ is 
a uniformly random orthonormal pair, orthogonal to $\bv$. Further, it
is independent of
$z_1$. By applying Lemma \ref{lemma:LLN_u}, we obtain that, for all
$t>0$ and all $n\ge n_0(t)$
\begin{align}
\prob\left(\frac{1}{n}\sum_{k=1}^n
  \big(\sqrt{n}\uper_{i,k}\big)^2\ge 1+t\right)&\le \frac{1}{n^9}\, ,\\
\prob\left(\frac{1}{n}\sum_{k=1}^n\big(\sqrt{n}\, \uper_{1,k}\big)
  \big(\sqrt{n}\uper_{i,k}\big)^2\ge t\right) & \le \frac{1}{n^9}\,
,\\
\prob\left(\frac{1}{n}\sum_{k=1}^n\big(\sqrt{n}\, \uper_{1,k}\big)^2
  \big(\sqrt{n}\uper_{i,k}\big)^2\ge 1+t\right)
\le \frac{1}{n^9}\, .
\end{align}
Using these in Eq.~(\ref{eq:ExpansionS2}) together with $z_1\in
[0,1]$, we get
\begin{align}
\prob\big(S_2\ge 1+t\big) \le \frac{1}{n^8}\, ,
\end{align}
for all $n\ge n_0(t)$. Using this together with Eq.~(\ref{eq:probS4})
in Eq.~(\ref{eq:uFu}) (with $t=\eps^2$), we obtain
that there exists an absolute  constant $C>$ such that, for all $n\ge
n_0(\eps)$
\begin{align}
\prob\Big(\<\bup_i,\bF\bup_i\>\ge 1-\frac{1}{2}\eps^2-C\eps^4\Big)\ge 1-
 \frac{1}{n^7}\, .
\end{align}
The claim  (\ref{eq:ClaimLemmaUi}) follows since $1-\eps^2/2\ge
\sqrt{1-\eps^2}$ for $\eps\in [0,1]$, and using Eq.~(\ref{eq:uFu-upFup}).
\end{proof}

We are now in position to prove Lemma \ref{lemma:LB_ui}.
\begin{proof}[Proof of Lemma \ref{lemma:LB_ui}]
Fix $i\in \{2,\dots, n\delta+1\}$. 
We claim that $\<\bu_i,\bF\bB\bF\bu_i\> \ge 2-2\eps^2- C\delta^{2/3}
  -C\eps^4$ holds with probability larger than $1- C/n^2$. 
In order to prove this, note that
\begin{align}
\<\bu_i,\bF\bB\bF\bu_i\>&= \xi_1\<\bu_i,\bF\bu_1\>^2+
\xi_i\<\bu_i,\bF\bu_i\>^2+\<\projp_{1,i}\bF\bu_i,\bB(\projp_{1,i}\bF\bu_i)\>\\
& \ge \xi_1\<\bu_i,\bF\bu_1\>^2+
\xi_i\<\bu_i,\bF\bu_i\>^2+\xi_n\big\|\projp_{1,i}\bF\bu_i\big\|_2^2\, .\label{eq:BoundFBFProj}
\end{align}
Let $\xi_*(\delta)$ be defined as in the previous section, namely as
the unique positive solution of Eq.~(\ref{eq:XiStarDef}).
(In particular, $\xi_*(\delta)\ge 2-C\,\delta^{2/3}$.)
Note that by \cite{knowles2013isotropic}[Theorem 2.7], we have, for all $n$ large enough
\begin{align}
\prob\big(\cE\big)& \ge
1-\frac{1}{n^{10}}\, ,\label{eq:BoundProbE}\\
\cE & = \Big\{\bB:\; \xi_1\ge \lambda+\lambda^{-1}-n^{-0.4},\; \xi_{n\delta+1}\ge
\xi_*(\delta)- n^{-0.4}, \; \xi_n\ge -2-n^{-0.4}\Big\}
\end{align}
On the event $\cE$, we have, by Eq.~(\ref{eq:BoundFBFProj}),
\begin{align}
\<\bu_i,\bF\bB\bF\bu_i\> & \ge (\lambda+\lambda^{-1}-n^{-0.4})\<\bu_i,\bF\bu_1\>^2+
(\xi_*(\delta)-
n^{-0.4})\<\bu_i,\bF\bu_i\>^2- (2+n^{-0.4}) \big\|\projp_{1,i}\bF\bu_i\big\|_2^2\\
& \ge (2-C\delta^{2/3}-
n^{-0.4})\<\bu_i,\bF\bu_i\>^2-3 \big\|\projp_{1,i}\bF\bu_i\big\|_2^2\, .
\end{align}
Using Eq.~(\ref{eq:BoundProbE}),  Lemma \ref{lemma:Perp}
and Lemma \ref{lemma:Ui} we obtain, for all $n\ge n_0(\eps)$
\begin{align}
\prob\Big(
\<\bu_i,\bF\bB\bF\bu_i\> \ge
\big(2-C\delta^{2/3}-n^{-0.4}\big)(1-\eps^2-C\eps^4)-3C^2\eps^4\Big)\ge
1-\frac{C}{n^4}\, .
\end{align}
The lemma follows by adjusting the constant $C$, and union bound over $i\in\{2,\dots,n\delta+1\}$.
\end{proof}
%
%**********
%
\section{Proof of Theorem \ref{thm:Estimation} (estimation)}

\subsection{A rounding lemma}

We will need the following rounding lemma, that is of independent interest.
While we state it for general expectations of random variables, we
will apply it to finite sums
(i.e. expectations with respect to random variables that take finitely many
values).
\begin{lemma}\label{lemma:RoundingPM1}
For $t\in \reals_{\ge 0}$, define $s_t:\reals\to \{+1,0,-1\}$ by
$s_t(x) = 1$ if $x\ge t$, $s_t(x) = -1$ if $x\le -t$, and $s_t(x) = 0$ otherwise. 

Let $X_0$, $Y$ be two random variables with $\prob(X_0=+1) = \prob(X_0=-1) = 1/2$,
$\E(X_0Y)\ge \eps>0$ and $\E(Y^2)=1$. Then, there exists $t_*$ (repending on the joint law of $X_0,Y$) such that
\begin{align}
\E\{X_0s_{t_*}(Y)\} \ge \frac{\eps^2}{4}\, .
\end{align}
\end{lemma}
\begin{proof}
Define $Z=X_0Y$. Then the assumptions translate into $\E(Z)\ge \eps$ and $\E(Z^2)=1$, while the claim is equivalent to 
$\E\{s_t(Z)\}\ge \eps^2/4$ (note indeed that $s_t(\,\cdot\,)$ is an odd function). 
Now we have
\begin{align}
\eps&\le \E(Z) = \int_{0}^{\infty}\big[\prob(Z\ge t)-\prob(Z\le -t)\big]\,\de t\\
& \le \int_{0}^{T}\big[\prob(Z\ge t)-\prob(Z\le -t)\big]\,\de t+\frac{1}{T}\int_{T}^{\infty}t\,\big[\prob(Z\ge t)+\prob(Z\le -t)\big]\,\de t\\
& \le \int_{0}^{T}\E\{s_t(Z)\}\,\de t+\frac{1}{T}\E\{Z^2\}\, .
\end{align}
Taking $T=2/\eps$, it follows that 
\begin{align}
\frac{1}{T}\int_{0}^{T}\E\{s_t(Z)\}\,\de t \ge \frac{\eps^2}{4}\, .
\end{align}
Since the average of $\E\{s_t(Z)\}$ over the interval $t\in[0,T]$ is at least $\eps^2/4$, then there must exists $t_*\in[0,T]$
such that $\E\{s_{t_*}(Z)\}\ge \eps^2/4$.
\end{proof}

\subsection{Proof of Theorem \ref{thm:Estimation}}

Throughout this appendix, the partition $V=S_1\cup S_2$ is fixed.
Note that $G_1\sim\sG(n,a'/n,b'/n)$, and $G_2\sim\sG(n,a'\delta_n/n,b'\delta_n/n)$, with $a' = a/(1+\delta_n)$, $b' = b/(1+\delta_n)$.  
For simplicity of notation, we will use $a$ instead of $a'$ and $b$ instead of $b'$. Note that this does not change the assumptions 
because it only implies a $o_n(1)$ shift in $a$, $b$. 
Also, $G_1$ and $G_2$ are dependent because they cannot share
edges. However, if they are sampled independently, 
they will share, with high probability,
only $O(1)$ edges. We will therefeore treat them as independent: the incurred error is negligible.

Setting, by definition,  the diagonal entries of $\bAc_{G_1}$ to be equal to
$\lambda\sqrt{d}$, we have
\begin{align}
\frac{1}{\sqrt{d}}\bAc_{G_1} = \frac{\lambda}{n}\bxz\bxz^{\sT}+ \bE\, ,
\end{align}
where $\bE=\bE^{\sT}$ has zero mean, $\E\{\bE\} = 0$, with $E_{ii}
=0$, and $(E_{ij})_{i<j}$ independent
\begin{align}
E_{ij} = \begin{cases}
\frac{1}{\sqrt{d}}\Big(1-\frac{d}{n}\Big) & \mbox{ with probability
  $p_{ij}$,}\\
-\frac{\sqrt{d}}{n} & \mbox{ with probability
  $1-p_{ij}$.}\\
\end{cases}
\end{align}
Here $p_{ij} = a/n$ if $\{i,j\}\subseteq S_1$ or $\{i,j\}\subseteq
S_2$, and $p_{ij} = b/n$ otherwise. 

Proceeding exactly as in the proof of Theorem \ref{thm:Approx}, we can
compare the SDP value for the matrix $\bE$, to the SDP value for a
Gaussian matrix. We obtain the following estimate, whose proof we
omit. 
\begin{lemma}
Let $\bE\in\reals^{n\times n}$ be the random matrix defined above,
with $d=(a+b)/2$, and  $\lambda = (a-b)/\sqrt{2(a+b)}$.
Let $\bW\sim \GOE(n)$ be a Gaussian random matrix with $(W_{ij})_{i<j}\sim_{i.i.d.}\normal(0,1/n)$.
Then, there exists $C=C(\lambda)$ such that, with probability at least
$1-C\, e^{-n/C}$, for all $n\ge n_0(a,b)$
\begin{align}
\left|\frac{1}{n}\SDP(\bE)-\frac{1}{n}\SDP(\bW)\right|&\le
  \frac{C\log d}{d^{1/10}}\, ,\label{eq:SDP_BerGOE}
\end{align}
Further $C(\lambda)$ is bounded over compact intervals $\lambda\in
[0,\lambda_{\rm max}]$
\end{lemma}
As a consequence of this lemma, and of Theorem \ref{thm:Gaussian}, we
have
\begin{align}
\frac{1}{n}\SDP(\bE)\le 2+ \frac{C\log d}{d^{1/10}}\, , \label{eq:SDPE}
\end{align}
with  probability at least $1-C\, e^{-n/C}$.

Consider then a maximizer  $\bX_*$ of the SDP
  (\ref{eq:SDP.DEF}), with $\bM=\bAc_{G_1}$. We have, by Theorem \ref{thm:Gaussian} and
  Theorem \ref{thm:Approx} (or, equivalently, by Theorem \ref{thm:SDP_Test})
\begin{align}
\frac{\lambda}{n^2}\<\bxz\bxz^{\sT},\bX_*\>+\frac{1}{n}\<\bE,\bX_*\>=\frac{1}{n\sqrt{d}}\SDP(\bAc)\ge
  2+\Delta(\eps)\, ,
\end{align}
for all $d\ge d_*(\eps)$, with probability at least $1-Ce^{-n/C}$. Using the bound (\ref{eq:SDPE}), this implies, for $\lambda$ bounded
and some $\Delta_2(\eps)>0$,
\begin{align}
\frac{1}{n^2}\sum_{i=1}^n\xi_i\<\bxz,\bv_i\>^2=\frac{1}{n^2}\<\bxz\bxz^{\sT},\bX_*\>\ge  \Delta_2(\eps)
\end{align}
Since $\bX_*\in\PSD_1(n)$, we have $\xi_i\ge 0$ and $\sum_{i=1}^n\xi_i=n$. Hence there exists $I_*\in [n]$ 
such that 
\begin{align}
\frac{1}{\sqrt{n}}\big|\<\bxz,\bv_{I_*}\>\big|\ge \sqrt{\Delta_2(\eps)}\, .
\end{align}
Assume, without loss of generality, that $\<\bxz,\bv_{I_*}\>\ge 0$. Applying Lemma \ref{lemma:RoundingPM1}
to the pair $(X_0,Y)$ with joint distribution $n^{-1}\sum_{j=1}^n\delta_{x_{0,j},v_{I_*,j}}$, we conclude that there exists 
$t_*\in\reals$ such that 
\begin{align}
\frac{1}{n}\<\bxz,s_{t_*}(\bv_{I_*})\>\ge \frac{\Delta_2(\eps)}{4}\, .
\end{align}
(Here $s_{t_*}(\,\cdot\,)$ is understood to be applied componentwise.)
Note that $s_{t_*}(\bv_{I_*})= \hbx^{(I_*,J_*)}$ for some index $J_*\in [n]$, whence 
\begin{align}
\frac{1}{n}\<\bxz,\hbx^{(I_*,J_*)}\>\ge \frac{\Delta_2(\eps)}{4}\, . \label{eq:LowerBoundEstimator}
\end{align}
In other words, at least one of the estimators $\{\hbx^{(i,j)}:\, i,j\in [n]\}$ has a good correlation with the ground truth. 
We are left to prove that step $(iv)$ in our algorithm does indeed select such a pair of indices $i,j$.
This follows from the following simple concentration lemma.
\begin{lemma}
There exists a constant $C=C(a,b)$ bounded for $a,b$ in bounded intervals, such that, for all $s\in [0,1]$.
\begin{align}
\prob\Big\{\max_{i,j\in [n]}\Big|\<\hbx^{(i,j)},\bAc_{G_2}\hbx^{(i,j)}\>-\frac{\lambda\sqrt{d}}{2n}\<\hbx^{(i,j)},\bxz\>^2+\frac{\lambda\sqrt{d}}{2}\Big|\ge s\sqrt{n}\Big\}\le C\, e^{-\sqrt{n}s^2/C}\, .
\end{align}  
\end{lemma} 
\begin{proof}
Throughout the proof, $C$ denotes a constant that might depend on $a,b$, bounded for $a,b$  in compact intervals.
For  any fixed vector $\hbx\in \{+1,0,-1\}^n$ we have, by Azuma-Hoeffding inequality
\begin{align}
\prob\Big\{\big|\<\hbx,\bAc_{G_2}\hbx\>-\E\<\hbx,\bAc_{G_2}\hbx\>\big|\ge s\; ;\;\; |E_2|\le m_2\Big\}\le 2\, e^{-s^2/8m_2}\, .
\end{align}  
However, by Chernoff bound, $m_2\le C\sqrt{n}$ with probability at least $1-Ce^{-n^{1/2}/C}$, whence,
for all $s\in [0,1]$, 
\begin{align}
\prob\Big\{\big|\<\hbx,\bAc_{G_2}\hbx\>-\E\<\hbx,\bAc_{G_2}\hbx\>\big|\ge s\sqrt{n}\Big\}\le C\, e^{-\sqrt{n}s^2/C}\, .
\end{align}  
On the other hand $\E\bAc_{G_2} = (a-b) (\bxz\bxz^{\sT}-\id)/(2n)$, whence 
\begin{align}
\prob\Big\{\Big|\<\hbx,\bAc_{G_2}\hbx\>-\frac{\lambda\sqrt{d}}{2n}\<\hbx,\bxz\>^2+\frac{\lambda\sqrt{d}}{2}\Big|\ge s\sqrt{n}\Big\}\le C\, e^{-\sqrt{n}s^2/C}\, .
\end{align}  
The claim follows by taking union bound over $i,j\in [n]$, since $\bx^{(i,j)}$ is independent of $G_1$.
\end{proof}

It follows from the last lemma, and Eq.~(\ref{eq:LowerBoundEstimator}) that
\begin{align}
\frac{1}{n}\max_{i,j\in [n]}\<\hbx^{(i,j)},\bAc_{G_2}\hbx^{(i,j)}\>\ge \frac{\lambda\sqrt{d}\Delta_2(\eps)^2}{64} \equiv\Delta_3(\eps)\, ,\label{eq:EventEst1}
\end{align}
with probability at least $1-Ce^{-n^{1/2}/C}$. Hence, again by the last lemma
\begin{align}
\max\Big\{\<\hbx^{(i,j)},\bAc_{G_2}\hbx^{(i,j)}\> :\, (i,j)\in [n] \frac{1}{n}|\<\hbx^{(i,j)},\bxz\>|\le \frac{\Delta_2(\eps)}{8}\Big\}\le 
\frac{\Delta_3(\eps)}{2}\label{eq:EventEst2}
\end{align}
with probability at least $1-Ce^{-n^{1/2}/C}$.
The claim follows since on the events (\ref{eq:EventEst1}) and (\ref{eq:EventEst2}) we necessarily have 
$|\<\hbx^{(I,J)},\bxz\>|\ge n\Delta_2(\eps)/8$.
%
%**********
%
\section{Proof of Corollary \ref{coro:Robustness} (robustness)}

Recall that $\bAc_G=\bA_C-(d/n)\bone\bone^{\sT}$ denotes the centered
adjacency matrix. 
If $G$ and $\tG$ differ in one edge, then
$|\SDP(\bAc_G)-\SDP(\bAc_{\tG})|\le 1$: a
complete proof of this simple fact is given in the proof of Lemma
\ref{lemma:ConcentrationSDP} below. The claim then follows immediately
since (using the coupling in the statement)
$|\SDP(\bAc_G)-\SDP(\bAc_{\tG})|=o(n)$ with high probability.

%
%*********************************************
%
\section{Proof of Theorem \ref{thm:SDP_Test_r} (testing $r>2$ communities)}

The proof is very similar to the one of Theorem \ref{thm:SDP_Test},
and we therefore limit ourself to an outline emphasizing the main
differences.
Throughout the proof we set
\begin{align}
d & = \frac{1}{r} \big[a+(r-1)b\big]\, ,\\
\lambda  &= \frac{a-b}{r\sqrt{d}} = \frac{a-b}{\sqrt{r(a+(r-1)b)}} \ge
1+\eps\, .
\end{align}
 Further, without loss of generality, we can assume
$\lambda\in [0,\lambda_{\rm max}]$ with $\lambda_{\rm max}>1$ fixed.
Also, the concentration lemma \ref{lemma:ConcentrationSDP} applies
unchanged to $\SDP(\bAc_G)$ for $G\sim\sG_r(n,a/n,b/n)$. It is
therefore  sufficient to check that the error probability vanishes as
$n\to\infty$. The exponentially decaying error rate follows.

Consider first the probability of a false positive (i.e. declaring
that $r$ communities are present when $G\sim\sG(n,d/n)$).
As for Theorem \ref{thm:SDP_Test}, we have
\begin{align}
\lim_{n\to\infty}\prob_0\big(T_r(G;\delta) =1 \big) &=
  \lim_{n\to\infty}\prob_0\Big(\frac{1}{n}\SDP(\bAc_G)\ge
  2(1+\delta)\sqrt{d} \Big)
 = 0\, .
\end{align}
where the last equality holds for any $d\ge d_0(\delta)$ by Theorem \ref{thm:Main}.

We are then left with the task of proving that the probability of
false negatives vanishes. This follows the same steps as for Theorem
\ref{thm:SDP_Test}. Namely: $(i)$ We approximate the value of
$\SDP(\bAc_G)$ for $G\sim\sG_r(n,a/n,b/n)$ by the value of the SDP
for a suitable deformed GOE model; $(ii)$ We analyze the deformed GOE
model.

The relevant deformed GOE random matrix is defined as follows.
Let $\bB_0(r)\in\reals^{n\times n}$ be given by
\begin{align}
B_0(r)_{i,j}= \begin{cases}
(r-1)/n & \mbox{ if $\{i,j\}\subseteq S_\ell$ for some $\ell\in[r]$,}\\
-1/n & \mbox{ otherwise.}
\end{cases}
\end{align}
Note that $\bB_0(r)$ has rank $(r-1)$, and all of its non-zero eigenvalues are equal to $\bB_0 = 1$. Hence $\bB_0 = \sum_{k=1}^{r-1}\bv_k\bv_k^{\sT}$, for
$\bv_1,\dots,\bv_{r-1}\in\reals^{n}$ an orthonormal set. 
We then let
\begin{align}
\bB(\lambda,r) = \lambda\, \bB_0(r) + \bW\, ,\label{eq:Bdefinition_r}
\end{align}
with $\bW\sim\GOE(n)$. 

We are now in position to state an analogue of the approximation
theorem \ref{thm:Approx}.
\begin{theorem}\label{thm:Approx_r}
Let $G\sim\sG_d(n,a/n,b/n)$, $d=(a+(r-1)b)/r$, and $\bAc_G = \bA_G-(d/n)\bone\bone^{\sT}$
be its centered adjacency matrix. Let $\lambda = (a-b)/(r\sqrt{d})$
and define $\bB = \bB(\lambda,r)$ to be the deformed GOE matrix in Eq.~(\ref{eq:Bdefinition_r}).
Then, there exists $C=C(\lambda,r)$ such that, with probability at least
$1-C\, e^{-n/C}$, for all $n\ge n_0(a,b,r)$
\begin{align}
\left|\frac{1}{n\sqrt{d}}\SDP(\bAc_G)-\frac{1}{n}\SDP(\bB(\lambda,r))\right|&\le
  \frac{C\log d}{d^{1/10}}\, ,\\
\left|\frac{1}{n\sqrt{d}}\SDP(-\bAc_G)-\frac{1}{n}\SDP(-\bB(\lambda,r))\right|&\le
  \frac{C\log d}{d^{1/10}}\, . 
\end{align}
Further $C(\lambda,r)$ is bounded over compact intervals $\lambda\in
[0,\lambda_{\rm max}]$
\end{theorem}
The proof of this theorem is exactly equal to the one of Theorem
\ref{thm:Approx_r}: $(i)$ We introduce a rank-constrained version of the
above SDP, and boud the error using the Grothendieck-type inequality
of Theorem \ref{thm:Gro}; $(ii)$ We introduce a `finite-temperature'
smoothing of this optimization problem, and bound the error using
Lemma \ref{lemma:ZeroTemperature}; $(iii)$ We use Lindeberg method as
in  Lemma \ref{lemma:Interpolation} to replace the centered adjacency
matrix $\bAc_G$ by the Gaussian model $\bB(\lambda,r)$. 
We will omit further details of this proof.

We then analyze the model $\bB(\lambda,r)$, and establish the
following analogue of Theorem \ref{thm:Gaussian}.
\begin{theorem}\label{thm:Gaussian_r}
Let $\bB=\bB(\lambda,r)\in\reals^{n\times n}$ be a symmetric matrix distributed
according to the model (\ref{eq:Bdefinition_r}), $r\ge 2$. 

If $\lambda>1$, then there exists $\Delta(\lambda,r)>0$ such that
$\SDP(\bB(\lambda,r))/n\ge 2+\Delta(\lambda,r)$ with probability
  converging to one as $n\to\infty$.
\end{theorem}
The proof of this result is very similar to the one of Theorem
\ref{thm:Gaussian}. We outline the main differences in Section
\ref{sec:ProofGaussianR}.

Armed with these theorems, we can now lower bound $\SDP(\bAc_G)$ for
$G\sim\sG_r(n,a/n,b/n)$. Namely, for $\lambda\ge 1+\eps$ we have, with
high probability,
\begin{align}
\frac{1}{n\sqrt{d}}\SDP(\bAc_G)&\ge \frac{1}{n}\SDP(\bB(\lambda,r)) -
                         \frac{1}{4}\Delta(1+\eps,r)\\
&\ge \frac{1}{n}\SDP(\bB(1+\eps,r)) -
                         \frac{1}{4}\Delta(1+\eps,r)\\
&\ge 2+\frac{3}{4}\Delta(1+\eps,r)\, .
\end{align}
We then conclude selecting $\delta_*(\eps) = \Delta(1+\eps)/2>0$,
as in the proof of Theorem \ref{thm:Gaussian}, see Eq.~(\ref{eq:ProbFalseNeg}). 

\subsection{Proof outline for Theorem \ref{thm:Gaussian_r}}
\label{sec:ProofGaussianR}

Throughout this section $\bB = \bB(\lambda,r)$ with  $\lambda\ge
1+\eps$ and $r\ge 2$ is defined as per Eq.~(\ref{eq:Bdefinition_r}).

As for the proof of Theorem \ref{thm:Gaussian}, the proof consists in 
constructing a suitable witness $\bX\in\PSD_1(n)$, and then lower
bounding the value $\<\bB,\bX\>$. We describe here the witness
construction since the lower bound on $\<\bB,\bX\>$ is analogous to
the one in the case $r=2$.
 
Denote by $(\bu_1,\xi_1)$, \dots,
$(\bu_n,\xi_n)$
denote the eigenpairs of $\bB$, namely
\begin{align}
\bB\bu_k = \xi_k\bu_k\, ,
\end{align}
where $\xi_1\ge \xi_2\ge \dots\ge \xi_n$. Our construction depends on
parameters $\eps,\delta>0$. Let $\bV\in\reals^{n\times (r-1)}$ be the
matrix whose $i$-th column is the eigenvector $\bu_i$(and hence
containing eigenvectors $\bu_1$, \dots, $\bu_{r-1}$), and
$\bU\in\reals^{n\times (n\delta)}$ be the matrix whose $i$-th
column is eigenvector $u_{r+i-1}$ (and hence
containing eigenvectors $\bu_r$, \dots, $\bu_{r+n\delta-1}$).

Define, with an abuse of notation $R:\reals^{r-1}\to \reals^{r-1}$ as
follows
\begin{align} 
R(\bx) \equiv\begin{cases}
\bx & \mbox{ if $\|bx\|_2\le 1$,}\\
\bx/\|\bx\|_2 & \mbox{ otherwise,}\\
\end{cases}
\end{align}
and define $\bPsi\in\reals^{n\times (r-1)}$ as 
$\bPsi\equiv R(\eps\sqrt{n}\, \bV)$ where $R(\,\cdot\,)$ is understood
to be applied row-by-row to   $\eps\sqrt{n}\, \bV\in\reals^{n\times
  (r-1)}$. Equivalently, for each $i\in [n]$, we have
\begin{align}
\bPsi^{\sT}\bfe_i = R(\eps\sqrt{n}\, \bV^{\sT}\bfe_i)\, .
\end{align}
We finally define a diagonal matrix $\bD\in\reals^{n\times n}$ with entries
\begin{align}
D_{ii} \equiv
\frac{\sqrt{1-\|\bPsi^{\sT}\bfe_i\|^2_2}}{\|\bU^{\sT}\bfe_i\|^2_2}
\end{align}
and construct the witness by setting
\begin{align}
\bX = \bPsi\bPsi^{\sT}+ \bD\bU\bU^{\sT}\bD\, .
\end{align}
We have $\bX\in\PSD_1(n)$ by construction. The proof that, with high probability,
$\<\bB,\bX\>/n\ge 2+\Delta(\lambda,r)$ follows the same steps as for
the case $r=2$, detailed in Appendix \ref{app:ProofGaussianB}.

\end{document}